\documentclass[11pt]{amsart}

\usepackage[english]{babel}
\usepackage[utf8]{inputenc}
\usepackage[T1]{fontenc}

\usepackage{amsmath}
\usepackage{amssymb}
\usepackage{mathtools}
\usepackage{amsthm}
\usepackage{physics}
\usepackage{graphicx}
\usepackage{subfig}
\usepackage{enumitem}
\usepackage{braket}
\usepackage{todonotes}
\usepackage{rotating}
\usepackage{microtype}

\usepackage[widespace]{fourier}

\usepackage{booktabs}
\usepackage{rotating}
\usepackage{tabularx}
\usepackage{multirow}

\usepackage{siunitx}

\usepackage{graphicx}
\usepackage{subfig}
\usepackage{wrapfig}
\usepackage{xcolor}
\usepackage{subfig}
\usepackage{tikz}
\usetikzlibrary{calc}
\usetikzlibrary{er}
\usetikzlibrary{cd}

\usepackage[pdftex,colorlinks=true]{hyperref} 
\hypersetup{urlcolor=blue, citecolor=red, linkcolor=blue}
\usepackage{cleveref}

\usepackage{etoolbox} 
\AtBeginEnvironment{pmatrix}{\everymath{\displaystyle}}
\AtBeginEnvironment{bmatrix}{\everymath{\displaystyle}}

\DeclareMathOperator{\Id}{Id}
\DeclareMathOperator{\ud}{d}

\DeclareMathOperator{\diag}{diag}
\DeclareMathOperator{\bigO}{O}

\newcommand{\Alg}{\mathfrak{A}}
\newcommand{\sC}{\mathsf{C}}

\numberwithin{equation}{section}

\newcommand{\Sph}{\mathbb{S}}
\newcommand{\R}{\mathbb{R}}
\newcommand{\Z}{\mathbb{Z}}
\newcommand{\N}{\mathbb{N}}
\newcommand{\Q}{\mathbb{Q}}

\newcommand{\Sl}{\mathfrak{sl}}

\renewcommand{\vec}[1]{\mathbf{#1}}

\renewcommand{\epsilon}{\varepsilon}
\renewcommand{\imath}{\mathrm{i}}

\newcommand{\vbeta}{\boldsymbol{\beta}}

\newcommand{\vxi}{\boldsymbol{\xi}}
\newcommand{\vpi}{\boldsymbol{\pi}}

\allowdisplaybreaks
\makeatletter 
\renewcommand{\pdv}[2]{\begingroup 
\@tempswafalse\toks@={}\count@=\z@ 
\@for\next:=#2\do 
{\expandafter\check@var\next\@nil
 \advance\count@\der@exp 
 \if@tempswa 
   \toks@=\expandafter{\the\toks@\,}%
 \else 
   \@tempswatrue 
 \fi 
 \toks@=\expandafter{\the\expandafter\toks@\expandafter\partial\der@var}}%
\frac{\partial\ifnum\count@=\@ne\else^{\number\count@}\fi#1}{\the\toks@}%
\endgroup} 
\def\check@var{\@ifstar{\mult@var}{\one@var}} 
\def\mult@var#1#2\@nil{\def\der@var{#2^{#1}}\def\der@exp{#1}} 
\def\one@var#1\@nil{\def\der@var{#1}\chardef\der@exp\@ne} 
\makeatother

\theoremstyle{plain}
\newtheorem{theorem}{Theorem}[section]

\newtheorem{corollary}[theorem]{Corollary}
\newtheorem{lemma}[theorem]{Lemma}
\theoremstyle{definition}
\newtheorem{definition}[theorem]{Definition}

\theoremstyle{remark}
\newtheorem{remark}[theorem]{Remark}
\theoremstyle{remark}

\newtheorem{conjecture}{Conjecture}
\newtheorem*{conjecture*}{Conjecture}

\usepackage[backend=biber,sorting=nyt,maxnames=99,giveninits=true,doi=false,citestyle=numeric-comp]{biblatex}
\title{The $\Sl_{2}(\R)$ coalgebra symmetry and the superintegrable discrete-time systems}
\author{Giorgio Gubbiotti}
\address[G. Gubbiotti]{Dipartimento di Matematica ``Federigo Enriques'',
Universit\`a degli Studi di Milano, Via C. Saldini 50, 20133 Milano, Italy
\& INFN Sezione Milano, Via G. Celoria 16, 20133 Milano, Italy}
\email{giorgio.gubbiotti@unimi.it}
\author{Danilo Latini}
\address[D. Latini]{School of Mathematics and Physics, The University
of Queensland, Brisbane, QLD 4072, Australia}
\email{d.latini@uq.edu.au}

\date{\today}

\begin{document}

\begin{abstract}
    In this paper, we classify all the  variational dis\-cre\-te-time
    systems in quasi-standard form in $N$ degrees of freedom admitting
    coalgebra symmetry with respect to the generic realisation of the
    Lie--Pois\-son algebra $\Sl_{2}(\R)$. This approach naturally yields
    several quasi-maximally and maximally superintegrable discrete-time
    systems, both known and new. We conjecture that this exhausts the
    (super)integrable cases associated with this algebraic construction.
\end{abstract}

\maketitle

\setcounter{tocdepth}{1}
\tableofcontents

\section{Introduction}
\label{sec:intro}

This paper is devoted to the classification and the study of a class
of discrete-time systems in $N$ degrees of freedom admitting coalgebra
symmetry with respect to the Lie--Pois\-son algebra $\Sl_{2}(\R)$. We
make use of the notion of coalgebra symmetry for discrete-time
systems we recently introduced in \cite{GLT_coalgebra}.  The main
outcome of this paper is that the coalgebra symmetry approach can be
fruitfully used to systematically produce superintegrable discrete-time
systems in an analogous way as its continuous counterpart introduced in
\cite{Ballesteros_et_al_1996, BallesterosRagnisco1998}. In particular,
within this paper we introduce several seemingly new discrete-time
systems, including an $N$ degrees of freedom maximally superintegrable
discretisation of the celebrated Smorodinski--Win\-ter\-nitz system
\cite{Fris1965,Evans1990}, one of the first maximally superintegrable
systems ever introduced \cite{MillerPostWinternitz2013R}.

To be more specific, we consider a class of discrete-time systems we
call the \emph{systems in quasi-standard form} which we define as the
discrete Euler--Lagrange equations (dEL) of the following class of
discrete Lagrangians (dLagrangian) \cite{Logan1973}:
\begin{equation}
    L = \sum_{k=1}^{N} \ell_{k}\left(q_{k}(t+h) q_{k}(t)\right)-
    V(\vec{q}(t)),
    \quad
    \vec{q}(t) = \left( q_{1}(t),\dots,q_{N}(t) \right).
    \label{eq:dLgen}
\end{equation}
Here $h>0$ is a (fixed) constant and $t\in h\Z$, while the functions
$q_{k}(t)$ are not supposed to be defined for all $t\in\R$. Following the
tradition of \cite{McLachlan1993, Suris1994Garnier, Suris1994Symmetric,
Suris1994inversesquare} we choose such notation to avoid confusing
double indexing in the formul\ae. Furthermore, we suppose that the
functions $\ell_{k}=\ell_{k}(\xi)$ are smooth and locally invertible
functions in a given open domain of $\R$.  That is, the dLagrangian
\eqref{eq:dLgen} is supposed to be a well-defined function on $U_{1}\times
U_{2}\subset (\R^{N})^{\times 2}$, where $U_{i}$, $i=1,2$ are open subsets of
$\R^{N}$. After we fix the explicit form of the functions $\ell_{k}$ we
will specify the form of the sets $U_{i}$, $i=1,2$, see \Cref{thm:rad}.

Computing the discrete Euler--Lagrange equations associated to the
dLagrangian \eqref{eq:dLgen} we have that a system in quasi-standard
form has the following explicit expression:
\begin{equation}
    \ell_{k}'\left( q_{k}(t+h)q_{k}(t) \right) q_{k}(t+h)
    +
    \ell_{k}'\left( q_{k}(t) q_{k}(t-h) \right) q_{k}(t-h)
    =
    \frac{\partial V(\vec{q}(t))}{\partial q_{k}(t)},
    \label{eq:dELgen}
\end{equation}
for $k=1,\dots,N$. The number $N$ is called the degrees of freedom of
the system, and will be denoted throughout the paper with the shorthand
notation d.o.f..
Because of the local invertibility assumptions on the functions $\ell_{k}$
we have that equation \eqref{eq:dELgen} is step-by-step solvable to
determine the next iterate, as it was done in \cite{Suris1994inversesquare}
in a particular case. In general, we will not show the solved equations to
avoid cumbersome expressions.

\begin{remark}
    Before going further we give a few remarks on the
    terminology we chose to adopt and the meaning of the parameters.
    \begin{itemize}
        \item The fixed parameter $h>0$ has the meaning of time step
            between two different stages of the evolution of the system.
            If $h\to0$ we have the so-called \emph{continuum limit}. To be 
            more precise, the limit $h\to0$ gives a (possibly trivial) 
            continuous-time system which needs to be discussed case-by-case. 
            Indeed, each discrete-time system
            need an ad-hoc scaling of the parameters. We will discuss the
            continuum limits of the discrete-time systems we found in this
            paper in \Cref{sec:sys}.
        \item We choose to call ``system in quasi-standard form''
            the systems \eqref{eq:dELgen} because if $\ell_{k}(\xi)=\xi$ 
            for all $k=1,\dots,N$ the system is in the so-called 
            \emph{standard form}, as defined in \cite{Suris1989, HietarintaJoshiNijhoff2016}.
    \end{itemize}

    \label{rem:quasistandard}
\end{remark}

From the general theory of discrete variational systems, see
\cite{Bruschietal1991,TranvanderKampQusipel2016} we have that they are
naturally symplectic. Indeed, we can introduce the following canonical
momenta:
\begin{equation}
    p_{k}(t)
    =
    \ell_{k}'\left( q_{k}(t) q_{k}(t-h) \right) q_{k}(t-h),
    \quad
    k=1,\dots,N,
    \label{eq:pkgen}
\end{equation}
and write the the dEL equations \eqref{eq:dELgen} in canonical form as:
\begin{subequations}
    \begin{gather}
        \ell_{k}'\left( q_{k}(t+h)q_{k}(t) \right) q_{k}(t+h)
        +
        p_{k}(t)
        =
        \frac{\partial V(\vec{q}(t))}{\partial q_{k}(t)},
        \label{eq:eqHamqgen}
        \\
        p_{k}(t+h) =
        \ell_{k}'\left( q_{k}(t+h) q_{k}(t) \right) q_{k}(t),
        \label{eq:eqHampgen}
    \end{gather}
    \label{eq:eqHamgen}%
\end{subequations}
with $k=1,\dots,N$. 
Since in the right hand side of equation \eqref{eq:eqHampgen} $\vec{q}(t+h)$ is present, 
the iteration step 
in equation \eqref{eq:eqHamgen} is intended as a two-step procedure.
That is, it must be accomplished in the following way:
\begin{equation}
    \left( \vec{q}(t),\vec{p}(t) \right)
    \xrightarrow{\eqref{eq:eqHamqgen}}
    \left( \vec{q}(t+h),\vec{p}(t) \right)
    \xrightarrow{\eqref{eq:eqHampgen}}
    \left( \vec{q}(t+h),\vec{p}(t+h) \right).
    \label{eq:twostepmap}
\end{equation}
We will denote the iteration from step $t$ to $t+h$ by $\varphi_{h}$.
As a map we will assume that $\varphi_{h}\colon U_{1}\times U_{2} \to U_{1}\times U_{2}$,
where the $U_{i}$ are properly chosen open subsets of $\R^{N}$.
Once the sets $U_{i}$, $i=1,2$, are fixed there is a one-to-one
correspondence between the discrete-time system in the form \eqref{eq:eqHamgen}
and its map form. So, in this paper we will use interchangeably the words
``discrete-time system'' and ``map''.

From the general theory of discrete-time variational system the
equations \eqref{eq:eqHamgen} preserve the 
canonical Poisson bracket:
\begin{equation}
    \left\{ q_{i}(t),q_{j}(t) \right\} = \left\{ p_{i}(t),p_{j}(t) \right\}=0,
    \quad
    \left\{ q_{i}(t),p_{j}(t) \right\} = \delta_{i,j},
    \label{eq:qpcan}
\end{equation}
see \cite{Veselov1991,Bruschietal1991,Suris1994inversesquare}. Here, by 
preservation of a Poisson bracket we mean that the following 
relation holds true:
\begin{subequations}
    \begin{gather}
        \left\{ q_{i}(t+h),q_{j}(t+h) \right\} =
        \left\{ q_{i}(t),q_{j}(t) \right\},
        \\
        \left\{ p_{i}(t+h),p_{j}(t+h) \right\}=
        \left\{ p_{i}(t),p_{j}(t) \right\},
        \\
        \left\{ q_{i}(t+h),p_{j}(t+h) \right\} = \left\{ q_{i}(t),p_{j}(t) \right\}.
    \end{gather}
\label{eq:qpcanpres}%
\end{subequations}
This implies that the map $\varphi_{h}$ is a \emph{Poisson map}.  In
particular, since the canonical Poisson bracket \eqref{eq:qpcan} has
maximal rank we have that the map $\varphi_{h}$ is a \emph{symplectic
map}.

From the general theory of integrable symplectic maps, we recall the
following definitions:
\begin{itemize}
    \item A symplectic map $\varphi_{h}\colon U_{1}\times U_{2} \to U_{1}\times U_{2}$ 
        possessing $N$ functionally independent invariants
        in involution with respect to a non-singular Poisson bracket
        for all values of $h>0$ is said to be a \emph{Liouville integrable map}.
    \item A Liouville integrable map $\varphi_{h}\colon U_{1}\times U_{2} \to U_{1}\times U_{2}$ 
        possessing $N+k$, with $k>0$, functionally independent invariants
        in involution with respect to a non-singular Poisson bracket
        for all values of $h>0$ is said to be a \emph{superintegrable map}.
    \item A Liouville integrable map $\varphi_{h}\colon U_{1}\times U_{2} \to U_{1}\times U_{2}$ 
        possessing $2N-2$ functionally independent invariants
        in involution with respect to a non-singular Poisson bracket
        for all values of $h>0$ is said to be a \emph{quasi-maximally superintegrable
        (QMS) map}.
    \item A Liouville integrable map $\varphi_{h}\colon U_{1}\times U_{2} \to U_{1}\times U_{2}$ 
        possessing $2N-1$ functionally independent invariants
        in involution with respect to a non-singular Poisson bracket
        for all values of $h>0$ is said to be a \emph{maximally superintegrable
        (QMS) map}.
\end{itemize}
These definitions naturally generalise to the case of Poisson maps:
a Poisson map $\varphi_{h}\colon U_{1}\times U_{2} \to U_{1}\times U_{2}$ 
is said to be \emph{Liouville--Poisson integrable} if, for all values of $h>0$, it admits $M-r$ functionally independent invariants
in involution, where $M$ is the number of equations and $2r$ is the rank
of the preserved Poisson structure.

\begin{remark}
    Differently from the review \cite{MillerPostWinternitz2013R} on
    continuous superintegrability we consider superintegrability to
    \emph{require} Liouville integrability, making it a stronger property.
    We decide to adopt this definition because as noticed already in
    \cite{GLT_coalgebra} not all discrete-time systems admitting more
    than $N$ non-commuting invariants are integrable with respect to 
    other commonly accepted notions of integrability for discrete-time
    systems, e.g. algebraic entropy \cite{BellonViallet1999}.
    \label{rem:supintvsliouvint}
\end{remark}

For a complete overview on the integrability of Poisson and
symplectic maps we refer to \cite{Bruschietal1991, Veselov1991,
TranvanderKampQusipel2016}, the review part of the thesis
\cite{TranPhDThesis}, and our previous paper \cite{GLT_coalgebra}.

Before recalling the construction of coalgebra symmetry for
discrete-time systems we make a final observation on continuum limits,
i.e. the limits as $h\to0^{+}$. Given a (super)integrable discrete-time
system, if a continuum limit is known then it is expected to be
(super)integrable as well. However, we observe that this does not
necessarily follow from the (super)integrability of the associated
discrete system. For instance, in \cite{Gubbiotti_lagr_add} were presented
several examples where in the continuum limit one invariant of a system
in two d.o.f. is lost, yet the continuous-time system possesses two
invariants. If both the discrete-time system for $h>0$ and its continuum
limit as $h\to 0$ are (super)integrable, we say that the discrete system
is a \emph{(super)integrable discretisation}.

We now briefly recall the definition of coalgebra symmetry
for discrete-time systems we introduced in \cite{GLT_coalgebra}. The
concept of coalgebra was formulated in quantum group theory
\cite{ChariPressley1994Book,Drinfeld1987}. Precisely, a \emph{coalgebra} is
a pair of objects $(\Alg,\Delta)$ where $\Alg$ is a unital, associative
algebra and $\Delta\colon \Alg \rightarrow \Alg\otimes \Alg$ is a
\emph{coassociative} map. 
That is, $\Delta$ satisfies the following condition:
\begin{equation}
    (\Delta \otimes \Id) \circ \Delta=(\Id \otimes \Delta) \circ \Delta
    \iff
    \begin{tikzpicture}[baseline={(0,0)},thick]
        \node (a1) at (0,0){$\Alg$};
        \node (a2) at ($({2*cos(30)},{2*sin(30)})$) {$\Alg\otimes\Alg$};
        \node (a3) at ($({2*cos(30)},{-2*sin(30)})$) {$\Alg\otimes\Alg$};
        \node (a4) at (4,0){$\Alg\otimes\Alg\otimes\Alg$};
        \draw[->] (a1) edge node[above left]{$\Delta$} (a2) (a2) edge node[above right]{$\Delta\otimes\Id$}(a4) ;
        \draw[->] (a1) edge node[below left]{$\Delta$} (a3) (a3) edge node[below right]{$\Id\otimes\Delta$} (a4);
    \end{tikzpicture}
    \label{eq:commdiag}
\end{equation}
and it is an algebra homomorphism from $\mathfrak{A}$ to  
$\mathfrak{A} \otimes \mathfrak{A}$:
\begin{equation}
    \Delta (X \cdot  Y) = \Delta (X)  \cdot \Delta (Y) \qquad \forall \, X, Y \in \mathfrak{A} \, .
    \label{eq:homalg}
\end{equation}
The map $\Delta$ is called the \emph{coproduct map}. When there is no
possible confusion on the coproduct map, it is customary to denote the
coalgebra simply by $\Alg$.

In \cite{GLT_coalgebra} we gave the following definition:
\begin{definition}
    A Poisson map $\varphi_{h}$ is said to possess the
    \emph{coalgebra symmetry} with respect to the Poisson coalgebra
    $(\Alg,\Delta)$ if for all $N\in\N$ the evolution of generators $A_i$, $i=1,\dots,K$ in
    a $N$ degrees of freedom realisation of the Poisson coalgebra is:
    \begin{enumerate}[label=(\roman*)]
        \item \emph{closed} in the Poisson coalgebra, that is: 
            \begin{equation}
                A_{i}(t+h)  = a_{i}(A_{1}(t),\dots,A_{K}(t)), 
                \quad i=1,\dots,K,
                \label{eq:closurerelation}
            \end{equation} 
            with $a_{i}\in \mathcal{C}^{\infty}(\Alg)$,
            \label{cond:closure}
        \item it is a Poisson map with respect to the Poisson algebra $\Alg$,
            i.e.:
            \begin{equation}
                \{ A_{i}(t+h) ,A_{j}(t+h)\} = \varphi_{h}\bigl(\{A_{i}(t),A_{j}(t) \}\bigr),
                \quad i,j=1,\dots,K,
                \label{eq:poissonalg}
            \end{equation}
            \label{cond:poisson}
        \item assuming that the Poisson algebra $\Alg$ admits 
            \emph{$r$ independent Casimir functions} $\{ C_{1}(t),\dots,C_{r}(t) \}$,
            these are preserved as invariants by the map $\varphi_{h}$, i.e.:
            \begin{equation}
                C_{i}\left(\varphi_{h}(t)\right) = C_{i}(t), \quad i=1,\dots,r.
                \label{eq:casmirpreservation}
            \end{equation}
            \label{cond:cas}
    \end{enumerate}
    \label{def:dcoalgebrasymmety}
\end{definition}
This definition allows us to provide an analogue of the
construction of the invariants for Hamiltonian systems given in
\cite{Ballesteros_et_al_1996,BallesterosRagnisco1998}, a result stated
in \cite[Theorem 3.3]{GLT_coalgebra}.

The plan of the paper is the following: in \Cref{sec:sl2} we remind
the main properties of the Lie--Poisson algebra $\Sl_{2}(\R)$ and
its generic symplectic realisation. In particular, we discuss the
construction of the corresponding left and right Casimir invariants. In
\Cref{sec:class} we classify all systems in quasi-standard form
admitting coalgebra symmetry with respect to the generic realisation
of the Lie--Poisson algebra $\Sl_{2}(\R)$.  This result is contained
in \Cref{thm:rad}. In \textsc{Section \ref{sec:poly}} we study the
Liouville--Poisson integrability of the dynamical system of the form
\eqref{eq:closurerelation} associated to the generic realisation
of the Lie--Poisson algebra $\Sl_{2}(\R)$.  To be more specific, we
impose the existence of an additional polynomial invariant and prove
in \Cref{thm:invariant} that it exists if the degree of the invariant
is 1, 2, or 3. For invariants of higher degree we conjecture, based
on the evidence obtained for degrees 4 and 5, that no polynomial
invariants exist. In \Cref{sec:sys} we identify the systems we found
in \Cref{sec:poly} in the explicit symplectic realisation. This
gives us several integrable systems, which we put in the context of
the literature.  In particular, we find a discrete-time maximally
superintegrable version of the Smorodinski--Winternitz system. We
note that the maximal superintegrability of this discrete-time model
follows from a peculiar construction of the $\Sl_{2}(\R)$ coalgebra
which is possible only in this specific case. We also find a QMS
reduction of the discrete-time Wojciechowski system, introduced in
\cite{Suris1994Symmetric}, and a generalisation of a $N$ degrees of
freedom autonomous discrete-time Painlev\'e I equation we obtained in
our earlier work \cite{GLT_coalgebra}. In \Cref{sec:concl} we provide
some concluding remarks and discuss the further possible developments.

\section{The $\Sl_{2}(\R)$ Lie--Poisson coalgebra}
\label{sec:sl2}

The three-dimensional Lie--Poisson algebra $\Sl_{2}(\R)$, spanned by the generators $\vec{J}:=\{J_{-}, J_{+} J_{3}\}$, is characterized by the following Lie--Poisson brackets:
\begin{equation}
    \left\{ J_{-},J_{+} \right\}=4J_{3},
    \quad
    \left\{ J_{3}, J_{+} \right\} =  2J_{+},
    \quad
    \left\{ J_{3}, J_{-} \right\} = -2J_{-},
    \label{eq:sl2comm}
\end{equation}
and it is endowed with the Casimir invariant:
\begin{equation}
    C = C(\vec{J})= J_{+} J_{-} - J_{3}^{2}.
    \label{eq:sl2cas}
\end{equation}
This Lie--Poisson algebra can be endowed with a ``natural'' coproduct map called the \emph{primitive coproduct} \cite{Tjin1992}. Its explicit action on the basis generators and the unit element is given by ($\mu=\pm,3$):
\begin{equation}
    \Delta (J_{\mu}) = J_{\mu} \otimes 1 + 1 \otimes J_{\mu},
    \quad
    \Delta(1) = 1 \otimes 1,
    \label{eq:prim}
\end{equation}
and extends to polynomial elements through the homomorphism property:
\begin{equation}
    \Delta (J_{\mu} J_{\nu}) =  \Delta (J_{\mu})  \Delta (J_{\nu})  , \quad \mu,\nu=\pm,3.  
    \label{eq:primextended}
\end{equation}
For example, the coproduct of the Casimir function \eqref{eq:sl2cas} is computed as:
\begin{equation}
    \begin{aligned}
        \Delta(C) &= \Delta (J_{+}J_{-}-J_{3}^2) =  \Delta(J_{+})\Delta(J_{-})- \Delta(J_{3})^2
        \\
        &= 
        (J_{+}\otimes 1 + 1 \otimes J_{+})(J_{-}\otimes 1 + 1\otimes J_{-})
        - (J_{3}\otimes 1 + 1\otimes J_{3})^2
        \\
    &= J_{+}J_{-}\otimes 1 + J_{+} \otimes J_{-}+ J_{-}\otimes J_{+} + 1\otimes J_{+}J_{-}
    - (J_{3}^{2}\otimes 1 + 2 J_{3}\otimes J_{3} + 1\otimes J_{3}^{2})
    \\
    &= (J_{+}J_{-}-J_{3}^{2})\otimes 1 
    + 1\otimes (J_{+}J_{-}-J_{3}^{2})
    + J_{+} \otimes J_{-}+ J_{-}\otimes J_{+} 
    - 2 J_{3}\otimes J_{3},
    \end{aligned}
    \label{eq:copcas}
\end{equation}
where we used the definition primitive coproduct \eqref{eq:prim} and 
the properties of the tensor products. From the definition of the Casimir invariant
\eqref{eq:sl2cas} we obtain the final result:
\begin{align}
  \Delta(C) = C \otimes 1+ 1\otimes C+J_+ \otimes J_-+J_- \otimes J_+-2 J_3 \otimes J_3 \, .
    \label{eq:copcasopened}
\end{align}
Note that the Casimir of $\Sl_{2}(\R)^{\otimes 2}$ genuinely
contains new information because it is not just two tensor copies of
the casimir $C$, but additional terms are present.

After giving this summary of the abstract properties related to the
$\Sl_{2}(\R)$ Lie--Poisson coalgebra we need a way to ``embed'' these
properties in the setting of (continuous or discrete-time) dynamical
systems. To this end, we use the concept of \emph{symplectic realisation
of a Lie--Poisson algebra} as expressed in the following definition:

\begin{definition}
    Assume we are given a Lie--Poisson algebra $\Alg$ generated by $\{ J_{1},\dots,J_{m} \}$, associated structure constants $c_{\mu,\nu}^{\rho}$  such as:
    \begin{equation}
        \left\{ J_{\mu},J_{\nu} \right\}_{\Alg} = \sum_{\rho=1}^{m} c_{\mu,\nu}^{\rho} J_{\rho},
        \quad
        \mu,\nu=1,\dots,m,
        \label{eq:structureconstants}
    \end{equation}
    and Casimir invariants $\left\{ C_{1},\dots,C_{r} \right\}$.
    A \emph{$N$ degrees of freedom symplectic realisation} of the 
    Lie--Poisson algebra $\Alg$ is a map
    \begin{equation}
        D \colon \Alg \to \Omega,
        \label{eq:repr}
    \end{equation}
    where $\Omega$ is an open subset of a symplectic manifold $(M,\omega)$ with
    $\dim M = 2N$, such that it preserves the commutation relations
    \eqref{eq:structureconstants}:
    \begin{equation}
        \left\{ D(J_{\mu}),D(J_{\nu})  \right\}_{\omega} =\sum_{\rho=1}^{m} c_{\mu,\nu}^{\rho}D(J_{\rho}).
        \quad
        \mu,\nu=1,\dots,m .
        \label{eq:reprstruct}
    \end{equation}
    Moreover, the realisation $D\colon \Alg\to\Omega$ is called
    \emph{generic} if the dimension of $M$ equals the number of generators
    of $\Alg$ minus the number of its Casimir functions, i.e.:
    \begin{equation}
        2N = m-r.
        \label{eq:scondition}
    \end{equation}
    \label{def:generic}
\end{definition}

\begin{remark}
    We remark that on a symplectic manifold $(M,\omega)$, by
    Darboux theorem \cite{Arnold1997} and its discrete-time analog
    \cite{Bruschietal1991} we can introduce two sets of canonically
    conjugated variables $(\vxi,\vpi) \in U_{\vxi}\times 
    U_{\vpi}\subseteq (\mathbb{R}^{N})^{\times2}$, 
    where $\vxi$ are the canonical coordinates, and $\vpi$ their corresponding
canonical momenta, defined on some open sets $U_{\vxi},U_{\vpi}\subseteq\R^{N}$, and 
        satisfying the canonical Poisson relations
    (analogous to \eqref{eq:qpcan}):
    \begin{equation}
        \left\{ \xi_{i}, \xi_{j}  \right\} = \left\{ \pi_{i}, \pi_{j} \right\} =0,
        \quad
        \left\{ \xi_{i},\pi_{j} \right\} = \delta_{i,j}.
        \label{eq:canrel}
    \end{equation}
    Throughout the paper, we will denote the discrete-time variables
    in lowercase Latin letters and the continuous time variables in
    capital Latin letters. The Greek letters $\vxi$ and $\vpi$ will denote
    variables that can be either continuous or discrete-time.  The
    difference between the discrete-time and the continuous-time cases
    is that in the continuous-time setting the dynamics is specified by a
    smooth function $H=H(\vxi,\vpi)$, while in the discrete-time setting
    by a symplectic map $\varphi_{h}\colon U_{\vxi}\times U_{\vpi}\to
    U_{\vxi}\times U_{\vpi}$, as discussed in the Introduction.
    \label{rem:localcoordinates}
\end{remark}

A one degree of freedom symplectic realisation of $\Sl_{2}(\R)$ on 
$(\xi_{1},\pi_{1})\in\R_{+}\times\R$ is given by:
\begin{equation}
    D(J_{+}) = \pi_{1}^{2} + \frac{b_{1}}{\xi_{1}^{2}},
    \quad
    D(J_{-}) = \xi_{1}^{2},
    \quad
    D(J_{3}) = \xi_{1}\pi_{1},
    \label{eq:sl2poissond1}
\end{equation}
for arbitrary $b_{1}\in\R$.
Indeed, from the canonical commutation relations \eqref{eq:canrel}
it is readily verified that:
\begin{equation}
    \left\{ D(J_{-}),D(J_{+}) \right\}=4D(J_{3}),
    \quad
    \left\{ D(J_{3}),D(J_{\pm}) \right\} =  \pm 2 D(J_{\pm}).
    \label{eq:sl2commreal}
\end{equation}
Moreover, this realisation is \emph{generic}, since
$N=1$, $m=3$ and $r=1$.
So, following the literature (see for example
\cite{Ballesteros_et_al2008PhysAtomNuclei,BlascoPhD}), we will call this
realisation \emph{the generic one d.o.f. symplectic realisation of
$\Sl_{2}(\R)$}. That said, we will \emph{identify} the generators with the
corresponding functions on the right hand side of \eqref{eq:sl2poissond1}. 
The image of the Casimir functions \eqref{eq:sl2cas} through the realisation 
\eqref{eq:sl2poissond1} reads:
\begin{equation}
    D(C(\vec{J}))=C(D(\vec{J}))=b_1 \in \mathbb{R} \, .
\end{equation}

Once a realisation has been constructed, the primitive coproduct
can be used to raise the number of degrees of freedom. For example,
from formula \eqref{eq:prim} we obtain the following functions on
$\R_{+}^{2}\times\R^2$:
\begin{equation}
    \begin{gathered}
    (D \otimes D)\Delta(J_{-}) = \xi_1^2+\xi_2^2,
    \quad
    (D \otimes D)\Delta(J_{+}) = \pi_1^2+\pi_2^2+\frac{b_1}{\xi_1^2}+\frac{b_2}{\xi_2^2}    
    \\
    (D \otimes D)\Delta(J_{3}) =\xi_1 \pi_1+\xi_2 \pi_2,
    \end{gathered}
    \label{newgenrepr}
\end{equation}
once two copies of the one degree of freedom symplectic realisation
\eqref{eq:sl2poissond1} are considered, each of them associated to the
corresponding site in the tensor product space $1 \otimes 2$
with associated free parameters $b_{i}\in\R$, $i=1,2$.

The functions \eqref{newgenrepr} provide a two degrees of
freedom symplectic realisation for the same Lie--Poisson algebra
$\mathfrak{sl}_2(\mathbb{R})$, now w.r.t. the canonical Poisson
bracket in $\R_{+}^{2}\times\R^2$. The crucial point is that
at this level the image of the Casimir element, which is now given by
\eqref{eq:copcasopened}, turns out to be:
\begin{equation}
    (D \otimes D)\Delta(C(\vec{J}))=
    C((D \otimes D)\Delta(\vec{J}))=
    (\xi_1 \pi_2-\xi_2 \pi_1)^2+b_1\frac{\xi_2^2}{\xi_1^2}+b_2\frac{\xi_1^2}{\xi_2^2}+b_1+b_2 \, ,
\label{cas2D}
 \end{equation}
that is, it is no longer a constant but a function.  Moreover, by
construction, this function Poisson commutes with the generators in the
two degrees of freedom symplectic realisation \eqref{newgenrepr}.

So, by applying the coproduct map iteratively, and extending its
definition through the following generalization of the coassociativity
property:
\begin{equation}
\Delta^{[N]}:=(\overbrace{\Id \otimes \dots \otimes \Id}^{N-2} \otimes \Delta^{[2]}) \circ \Delta^{[N-1]} = (\Delta^{[2]} \otimes \overbrace{\Id \otimes \dots \otimes \Id}^{N-2}) \circ \Delta^{[N-1]}
\label{copext}
\end{equation}
where $\Delta^{[1]}:=\Id$ and $\Delta^{[2]}:=\Delta$, one ends up with
the $N$ degrees of freedom symplectic realisation:
\begin{equation}
    J_{+} = \sum_{k=1}^{N} \left( \pi_{k}^{2} + \frac{b_{k}}{\xi_{k}^{2}} \right),
    \quad
    J_{-} = \sum_{k=1}^{N} \xi_{k}^{2},
    \quad
    J_{3} = \sum_{k=1}^{N} \xi_{k}\pi_{k},
    \label{eq:sl2poisson}
\end{equation}
where the canonical variables are defined on $\R_{+}^{N}\times\R^{N}$,
and $b_{k}\in \R$, $k=1,\dots,N$ are $N$ associated arbitrary constants.
In equation \eqref{eq:sl2poisson} and the following, we omit
the symbol $(D \otimes \dots \otimes D)$, because we will not consider
different realisations. At this level, the crucial fact is that a total
number of $(2N-3)$  \emph{left} and \emph{right} Casimir invariants
can be obtained from the left and right embedding of the $m$-th order
($2 \leq m \leq N$) coproduct on the Casimir function \eqref{eq:sl2cas}:
\begin{equation}
\mathsf{C}^{[m]}:=\Delta^{[m]}(C)\otimes \overbrace{ 1 \otimes \dots \otimes 1}^{N-m} , \quad \mathsf{C}_{[m]}:=\overbrace{1 \otimes \dots \otimes 1}^{N-m} \otimes \Delta^{[m]}(C)  \, .
\end{equation}
The image of these elements under the $N$ d.o.f. symplectic realisation \eqref{eq:sl2poisson} reads as:
\begin{subequations}
    \begin{align}
        \mathsf{C}^{[m]} &
        =\hskip 0.5cm
       \sum_{1 \leq i <j}^m\hskip 0.35cm
        \left[L_{i,j}^2+b_i \frac{\xi_j^2}{\xi_i^2}+b_j \frac{\xi_i^2}{\xi_j^2}\right]
        +\hskip 0.5cm\sum_{j=1}^m  b_j 
        \qquad m=2, \dots, N,
        \label{eq:Cleftm}
        \\
        \mathsf{C}_{[m]} &
        =
       \sum_{N-m+1 \leq i <j}^N
        \left[L_{i,j}^2+b_i \frac{\xi_j^2}{\xi_i^2}+b_j \frac{\xi_i^2}{\xi_j^2}\right]
        +\sum_{j=N-m+1}^N b_j 
        \quad m=2, \dots, N,
        \label{eq:Crightm}
    \end{align}
    \label{eq:Cm}%
\end{subequations}
where we indicated the $N(N-1)/2$ rotation generators as:
\begin{equation}
    L_{i,j} := \xi_i \pi_j-\xi_j \pi_i \, .
    \label{eq:Lij}
\end{equation}
These $(2N-3)$ quadratic (in the momenta) functions Poisson commute
with the generators \eqref{eq:sl2poisson} by construction.  Moreover,
they turn out to be functionally independent.

\begin{remark}
    We remark that if we restrict to the case $N=2$ the left and right 
    Casimir functions collapse to the same expression:
    \begin{equation}
        \mathsf{C}^{[2]}=\mathsf{C}_{[2]}=\Delta^{[2]}(C)
    \label{2dcase}
    \end{equation}
    which is nothing but \eqref{eq:copcasopened}, the latter leading to 
    the invariant \eqref{cas2D} at a fixed realisation, as expected. 
    This extends to any $N$, in fact for $m=N$ the two expressions collapse to:
    \begin{equation}
        \mathsf{C}^{[N]}=\mathsf{C}_{[N]}=\Delta^{[N]}(C) ,
    \label{Ndcase}
    \end{equation}
    where the action of the $N$th-order coproduct is given by \eqref{copext}.
    This is why formul\ae\ \eqref{eq:Cm} give us $2N-3$ functionally
    independent invariants and not $2N-2$.
    \label{rem:Nthcasimir}
\end{remark}

So, if the variables $(\vxi,\vpi):=(\vec{Q},\vec{P})\in U_{\vec{Q}}\times U_{\vec{P}}\subset\R_{+}^{N}\times\R^{N}$ are
\emph{continuous}, we can conclude that the family of Hamiltonian systems:
\begin{equation}
    h=h(J_-, J_+, J_3)=h\left(\vec{Q}^2,  \vec{P}^{2} +\sum_{k=1}^{N}  \frac{b_{k}}{Q_{k}^{2}} , \vec{Q} \cdot \vec{P} \right)
    \label{eq:hN}
\end{equation}
where $h$ is any smooth function of the generators of the Lie--Poisson
algebra $\mathfrak{sl}_2(\mathbb{R})$, is QMS. This is because the
above Hamiltonian is automatically endowed with the $2N-3$ functionally
independent invariants \eqref{eq:Cm}.
On the other hand, in the discrete-time setting, there is no exact equivalent
of the Hamiltonian function and the notion of coalgebra symmetry 
is replaced by \Cref{def:dcoalgebrasymmety}.
Thus, the problem of characterising discrete-time (symplectic) systems
admitting $\Sl_{2}(\R)$ as a hidden coalgebra symmetry is much less trivial. In the case of systems in quasi-standard form, this problem is solved in \Cref{thm:rad}.
Moreover, due to the absence of the Hamiltonian, these systems are not
naturally born QMS, and in fact not even Liouville integrable.
The problem of integrability is tackled in \Cref{sec:poly}.

\begin{remark}
    We remark that the symplectic realisation with $b_i=0$, for $i=1,
    \dots, N$, is connected to radially symmetric systems.  In this particular case, the left and right Casimir invariants are:
    \begin{equation}
            \mathsf{C}^{[m]} 
            =
           \sum_{1 \leq i <j}^m
           L_{i,j}^2,
            \qquad 
            \mathsf{C}_{[m]} 
            =
           \sum_{N-m+1 \leq i <j}^N
            L_{i,j}^2 
            \qquad m=2, \dots, N,
            \label{eq:casrad}
    \end{equation}
    which are nothing but the Casimir invariants associated with
    rotation subalgebras $\mathfrak{so}(m) \subseteq \mathfrak{so}(N)$.
    In the continuous setting, the Hamiltonian function \eqref{eq:hN}
    is given by:
    \begin{equation}
        h=h(J_-, J_+, J_3)=h\left(\vec{Q}^2,  \vec{P}^{2}, \vec{Q} \cdot \vec{P}
        \right).
        \label{eq:radham}
    \end{equation}
    As expected, rotational symmetry is sufficient to provide quasi-maximal
    superintegrability. Notice that if we restrict to natural Hamiltonian
    systems defined in Euclidean $N$-space, i.e. we take:
    \begin{equation}
        h(J_-,J_+,J_3)=J_++V(J_-)
        \label{eq:radhamkosc}
    \end{equation}
    then as a consequence of Bertrand's Theorem \cite{Bertrand1873,
    Arnold1997}, only two MS subcases arise. Namely, the Harmonic and Kepler-Coulomb 
    (KC) systems, with the corresponding potentials given by:
    \begin{equation}
        V(J_-) = \alpha J_-   \quad \text{and} \quad V(J_-)=-\frac{\alpha}{\sqrt{J_-}} \, ,
    \label{potrad}
    \end{equation} 
    respectively. In \cite[Proposition 4.2]{GLT_coalgebra} it was proved
    that discrete-time radial systems in standard form admit $\Sl_{2}(\R)$ coalgebra.
    In the present paper in \Cref{cor:radialsystems} the converse is proved.
    \label{rem:radial}
\end{remark}

In general, the presence of non-central terms has the effect of breaking
the radial symmetry, but by preserving quasi-maximal superintegrability,
which is kept thanks to the existence of the new integrals obtained
through the image of the left and right Casimir invariants under
the realisation \eqref{eq:sl2poisson}. Let us conclude the Section
by mentioning that the same choice of potentials \eqref{potrad},
but now in terms of the new realisation involving non-central terms,
would result respectively in the $N$ d.o.f. Smorodinsky-Winternitz system
\cite{Evans1990} and the $N$ d.o.f. generalized KC system, the latter being
the $N$ d.o.f. generalization of the (fourth-order) superintegrable Hamiltonian
introduced in \cite{VerrierEvans2008}.

Although the many advances pursued in the framework of superintegrable
systems with an arbitrary number of d.o.f. using the method described
above, which has also been applied to the analysis of many other
Lie--Poisson coalgebras \cite{Ballesteros_2008, BlascoPhD}, the problem
of finding MS subcases is usually left  open. This is due to the fact
that additional integrals, not directly obtainable following this
algebraic approach, may arise. From this perspective, even for the
$\mathfrak{sl}_2(\mathbb{R})$ Lie--Poisson algebra, we have recalled
how it is possible to reach quasi-maximal superintegrability at most,
as a total number of $2N-3$ functionally independent integrals can be
constructed, besides the Hamiltonian, from the left and right Casimir
invariants. Thus, for MS subcases, the search for an additional missing
functionally independent constant would be required.

\section{Classification results}
\label{sec:class}

We state and prove the following classification result:
\begin{theorem}
    A system in quasi-standard form \eqref{eq:dELgen} possesses coalgebra
    symmetry with respect to $\Sl_{2}(\R)$ if and only if:
    \begin{equation}
        \ell_{k} (\xi) = \int^{\xi} \sqrt{1 - \frac{b_{k}}{\eta^{2}}}\ud \eta,
        \;\;
        b_{k}\in\R,
        \quad
        V = V(\vec{q}^{2}).
        \label{eq:Vrad}
    \end{equation}
    This implies that the dLagrangian and the associated pairs of
    canonical variables are well-defined in the open set
    $U=(\R_{+}^{N})^{\times 2}\subset (\R^{N})^{\times 2}$, or an open subset of it
    depending on the form of the function $V=V(\chi)$.
    \label{thm:rad}
\end{theorem}

\begin{remark}
    The ``if'' part of this theorem is a generalisation of 
    \cite[Proposition 4.2]{GLT_coalgebra}, where it is proved that radial
    systems in standard form admit coalgebra symmetry.
    The converse will be explicitly stated in \Cref{cor:radialsystems}.
    \label{rem:convtheorem}
\end{remark}

\begin{proof}
    Consider the $N$ d.o.f. realisation of $\Sl_{2}(\R)$ given
    in equation \eqref{eq:sl2poisson} with respect to the discrete
    canonical variables $(\vec{q}(t),\vec{p}(t))$. Then, the evolution
    of the $J_{+}(t)$ under the system \eqref{eq:eqHamgen} is:
    \begin{equation}
        J_{+}(t+h) = \sum_{k=1}^{N} 
        \left[ q_{k}^{2}(t) \left(\ell_{k}'(q_{k}(t)q_{k}(t+h))\right)^{2}
        +\frac{b_{k}}{q_{k}^{2}(t+h)} \right].
        \label{eq:Jpstart}
    \end{equation}
    The right-hand side needs to be independent from $\vec{q}_{k}(t+h)$.
    To ensure this independence, we can differentiate equation \eqref{eq:Jpstart}
    with respect to $q_{l}(t+h)$ and put the result identically equal to zero:
    \begin{equation}
        2 q_{l}^{3}(t) \ell_{l}'(q_{l}(t)q_{l}(t+h))\ell_{l}''(q_{l}(t)q_{l}(t+h))
        -2\frac{b_{l}}{q_{l}^{3}(t+h)} = 0, \quad l=1,\dots,N.
        \label{eq:Jpstartdiff}
    \end{equation}
    Interpreting the equation in terms of the variable $\xi_{l} = q_{l}(t)q_{l}(t+h)$,
    we can integrate equation \eqref{eq:Jpstartdiff} to:
    \begin{equation}
        \left(\ell_{l}'(\xi_{l})\right)^{2}
        =C_{l}- \frac{b_{l}}{\xi_{l}^{2}}, \quad l=1,\dots,N.
        \label{eq:Jpstartdiffsol}
    \end{equation}
    Plugging back into equation \eqref{eq:Jpstart} we obtain:
    \begin{equation}
        J_{+}(t+h) = \sum_{k=1}^{N} C_{k}q_{k}^{2}(t).
        \label{eq:Jpstartbis}
    \end{equation}
    Imposing that the right-hand side of the equation is in 
    $\mathcal{C}^{\infty}(\Sl_{2}(\R))$ we easily obtain that $C_{k}=C$
    for all $k$.
    Note that, when integrating we obtain a condition of the following form:
    \begin{equation}
       \ell_{l}(\xi_{l})
        =\int^{\xi_{l}}\sqrt{C- \frac{b_{l}}{\eta^{2}}}\ud \eta+D_{l}, \quad l=1,\dots,N,
        \label{eq:Jpstartdiffsolfin}
    \end{equation}
    but the additive constant $D_{l}$ is inessential to equations of motion
    \eqref{eq:dELgen} and can be safely put to zero.
    Then, using the scaling:
    \begin{equation}
        b_{l} \to C b_{l},
        \quad
        V(\vec{q}) \to \sqrt{C} V(\vec{q}),
        \label{eq:scalingC}
    \end{equation}
    we see that we can choose without loss of generality $C=1$.
    This proves the first part of the theorem, represented by the
    first identity in \eqref{eq:Vrad}. Moreover, from this
    follows that the dLagrangian and the associated pairs of canonical
    variables are well-defined in the open set $U=(\R_{+}^{N})^{\times
    2}\subset (\R^{N})^{\times 2}$ or an open subset of it depending on the form of
    the function $V=V(\vec{q}(t))$.

    We fixed the form of the coefficients $\ell_{k}$, so now we proceed 
    to derive the full system in order to check the conditions of
    \Cref{def:dcoalgebrasymmety}. By direct computations we have:
    \begin{subequations}
        \begin{align}
            J_{+}(t+h) &= J_{-}(t),
            \label{eq:Jpevolcond}
            \\
            J_{-}(t+h) &= J_{+}(t) - 2 \vec{p} \cdot \grad V + \left( \grad V \right)^{2},
            \label{Jmevolcond}
            \\
            J_{3}(t+h) &= -J_{3}(t) + \vec{q} \cdot \grad V.
            \label{J3evolcond}
        \end{align}
        \label{eq:sl2evolcond}%
    \end{subequations}
    While it is possible to prove that the system \eqref{eq:sl2evolcond}
    satisfies condition \ref{cond:poisson}, is it clear that 
    conditions \ref{cond:closure} and \ref{cond:cas} are not satisfied
    in general. As outlined in the last section of \cite{GLT_coalgebra}
    we start by imposing that condition \ref{cond:cas} holds.
    Evaluating the translation of the Casimir function \eqref{eq:sl2cas}
    we have:
    \begin{equation}
        C(t+h) = C(t) -
        \left[ 2 \vec{q}^{2} (\vec{p}\cdot \grad V) 
        - 2 (\vec{q}\cdot\vec{p})(\vec{q}\cdot \grad V) - 
        \vec{q}^{2} \left( \grad V \right)^{2} - \left( \vec{q}\cdot \grad V \right)^{2}
        \right] .
        \label{eq:Ccond}
    \end{equation}
    To have the preservation of the Casimir functions, we impose that the term
    in square brackets is identically zero:
    \begin{equation}
        2 \vec{q}^{2} (\vec{p}\cdot \grad V) 
        - 2 (\vec{q}\cdot\vec{p})(\vec{q}\cdot \grad V) - 
        \vec{q}^{2} \left( \grad V \right)^{2} - \left( \vec{q}\cdot \grad V \right)^{2}
        =0.
        \label{eq:Rcond}
    \end{equation}
    Since $V$ does not depend on $\vec{p}$ we can take coefficients with
    respect to it. In this way we obtain that $V$ has to satisfy the following system
    of equations:
    \begin{subequations}
        \begin{align}
            p_{i}^{1} &\colon 
            \vec{q}^{2} \frac{\partial V}{\partial q_{i}} - q_{i} \,\vec{q}\cdot \grad V = 0,
            \label{eq:pi1}
            \\
            p_{i}^{0} &\colon 
            \vec{q}^{2} \left( \grad V \right)^{2} - \left(\vec{q}\cdot \grad V\right)^{2} = 0.
            \label{eq:pi0}
        \end{align}%
        \label{eq:picoeffs}
    \end{subequations}
    Apply Lagrange's identity \cite{Weisstein2003}: 
    \begin{equation} 
        \sum_{k=1}^{N} v_{k}^{2} \sum_{k=1}^{N} w_{k}^{2} - \left( \sum_{k=1}^{N} v_{k} w_{k} \right)^2
        =
        \sum_{i=1}^{N-1}\sum_{j=i+1}^{N} \left( v_{i}w_{j}-v_{j}w_{i} \right)^{2}
        \label{eq:lagrangeid}
    \end{equation}
    to equation \eqref{eq:pi0} to obtain:
    \begin{equation} 
        \vec{q}^{2} \left( \grad V \right)^{2} - \left(\vec{q}\cdot \grad V\right)^{2}
        =
        \sum_{i=1}^{N-1}\sum_{j=i+1}^{N} 
        \left( q_{i}\frac{\partial V}{\partial q_{j}}-q_{j}\frac{\partial V}{\partial q_{i}} \right)^{2} = 0.
        \label{eq:lagrangeidpi0}
    \end{equation}
    This readily implies:
    \begin{equation}
        q_{i} \frac{\partial V}{\partial q_{j}} - q_{j} \frac{\partial V}{\partial q_{i}} =0,
        \quad i=1,\dots,N,\quad j=i+1,\dots,N.
        \label{eq:Vrot}
    \end{equation}
    Not all equations in \eqref{eq:Vrot} are independent: one can
    only consider the subset with $i=1$ and $j=2,\dots,N$.
    The solution of this system of $N-1$ partial differential equations 
    is given by $V=V(\vec{q}^{2})$, as in formula \eqref{eq:Vrad}.
    The introduction of $V=V(\vec{q}^{2})$ makes equation \eqref{eq:pi1}
    vanish identically meaning that condition \ref{cond:cas} is satisfied.
    We observe then that further restrictions on the domain of the symplectic
    map can only come from the singularities of the function $V=V(\chi)$.
    Finally, the $\Sl_{2}(\R)$ associated dynamical system has the following
    form:
    \begin{subequations}
        \begin{align}
            J_{+}(t+h) &= J_{-}(t),
            \label{eq:Jpevol}
            \\
            J_{-}(t+h) &= J_{+}(t) - 4 J_{3}(t) V'\left( J_{-}(t) \right)
            +4 J_{-} \left[ V'\left( J_{-}(t) \right) \right]^{2},
            \label{Jmevol}
            \\
            J_{3}(t+h) &= -J_{3}(t) + 2 J_{-}V'\left( J_{-}(t) \right).
            \label{J3evol}
        \end{align}
        \label{eq:sl2evol}%
    \end{subequations}
    Hence, condition \ref{cond:closure} is satisfied, and this ends
    the proof of the theorem.
\end{proof}

From Theorem \ref{thm:rad} we have that the most general form of the
$\Sl_{2}(\R)$ coalgebrically symmetric Lagrangian is:
\begin{equation}
    L = \sum_{k=1}^{N} 
    \int^{q_{k}(t)q_{k}(t+h)}\sqrt{1-\frac{b_{k}}{\eta^{2}}}\ud\eta
    -V(\vec{q}^{2}(t)),
    \label{eq:dLsl2symm}
\end{equation}
whose associated discrete Euler--Lagrange equations are:
\begin{equation}
    q_{k}(t+h)
    \sqrt{1-\frac{b_{k}}{q_{k}^{2}(t)q_{k}^{2}(t+h)}}
    +
    q_{k}(t-h)
    \sqrt{1-\frac{b_{k}}{q_{k}^{2}(t)q_{k}^{2}(t-h)}}
    =
    2 V'(\vec{q}^{2}(t))q_{k}.
    \label{eq:dELsl2}
\end{equation}
From formula \eqref{eq:pkgen} the symplectic form of the system is:
\begin{subequations}
    \begin{gather}
        q_{k}(t+h)
        \sqrt{1-\frac{b_{k}}{q_{k}^{2}(t)q_{k}^{2}(t+h)}}
        +
        p_{k}(t)
        =
        2 V'(\vec{q}^{2}(t))q_{k},
        \label{eq:eqHamqsl2}
        \\
        p_{k}(t+h) =
        q_{k}(t)\sqrt{1-\frac{b_{k}}{q_{k}^{2}(t)q_{k}^{2}(t+h)}}.
        \label{eq:eqHampsl2}
    \end{gather}
    \label{eq:eqHamsl2}
\end{subequations}

We conclude this Section with a corollary which represents the
inverse of \cite[Proposition 4.2]{GLT_coalgebra}:

\begin{corollary}
    A system in standard form
    \begin{equation}
        L = \sum_{i=1}^{N} q_{i}(t+h)q_{i}(t) + V(\vec{q}),
        \label{eq:standardform}
    \end{equation}
    possesses coalgebra symmetry with respect to $\Sl_{2}(\R)$ if and only if
    it is \emph{radially symmetric}, i.e. $V = V(\vec{q}^{2})$.
    \label{cor:radialsystems}
\end{corollary}

\begin{proof}
    Observe that, from equation \eqref{eq:dLsl2symm}, if $b_{k}=0$ for all
    $k=1,\dots,N$ the corresponding system becomes in standard form. So,
    since \Cref{thm:rad} is a necessary and sufficient condition the
    statement follows.
\end{proof}

\begin{remark}
    We remark that the  dynamical system associated with the $\Sl_{2}(\R)$ 
    \eqref{eq:sl2evol} is independent from the value of the constants
    $b_{k}$. So, from the coalgebra point of view systems
    in standard form, that is $b_{k}=0$ for all $k=1,\dots, N$, and in quasi-standard form can be considered in the same way: the
    (super)integrability of a system in standard form can be transferred
    to a system in quasi-standard form and vice-versa.
    \label{rem:bk0}
\end{remark}

\section{Polynomial invariants of the associated dynamical system}
\label{sec:poly}

As discussed in \cite{GLT_coalgebra} in the discrete-time setting
the existence of a rank one coalgebra symmetry alone it  is
not enough to guarantee the Liouville integrability of the
underlying discrete-time system. This is different from the
continuum setting, when this has been proven to be true, see
\cite{Ballesteros_et_al2008PhysAtomNuclei,Ballesteros_et_al2009}.

So, in this section, we will discuss the integrability of the system
\eqref{eq:sl2evol} assuming the existence of an additional polynomial
invariant:
\begin{equation}
    I_{d} = \sum_{1\leq i+j+k \leq d} a_{i,j,k}J_{+}^{i}J_{-}^{j}J_{3}^{k}.
    \label{eq:Idgen}
\end{equation}
Then we have the following result:
\begin{theorem}
    If $1\leq d \leq 3$ the system \eqref{eq:sl2evol} admits a polynomial 
    invariant of the form \eqref{eq:Idgen}, and in such cases we have:    
    \begin{subequations}
        \begin{align}
            I_{1} &= J_{+} + J_{-} - \kappa J_{3}, 
            \label{eq:I1}
            \\
            I_{2} &= \lambda_{3}\left(J_{+}+J_{-}\right)-\lambda_{1} J_{3}-\lambda_{2} J_{3}^2,
            \label{eq:I2}
            \\
            I_{3} &=
            \left(J_{+}+J_{-}\mp 2 J_3\right) 
            \left(\tau^2-2 \tau J_{3}+J_{+} J_{-}\right)
            \label{eq:I3}
        \end{align}
        \label{eq:I123}%
    \end{subequations}
    with the corresponding functions $f(\xi) = V'(\xi)$:
    \begin{subequations}
        \begin{align}
            f_{1} &= \frac{\kappa}{2},
            \label{eq:f1}
            \\
            f_{2} &= \frac{\lambda_{1}}{2} \frac{1}{\lambda_{3}-\lambda_{2}\xi},
            \label{eq:f2}
            \\
            f_{3} &= \pm \frac{1}{2} + \frac{\tau}{2} \frac{1}{\xi}.
            \label{eq:f3}
        \end{align}
        \label{eq:f123}%
    \end{subequations}
    \label{thm:invariant}
\end{theorem}

\begin{remark}
    We remark that through degeneration of parameters in \eqref{eq:f123}
    the functions $f_{2}$ and $f_{3}$ are both connected to the
    function $f_{1}$. Indeed, if $\lambda_{2}\to 0$ then $f_{2}$
    degenerates to $f_{1}$, provided we make the identification
    $\kappa=\lambda_{1}/\lambda_{3}$. Similarly, if $\tau\to0$ then
    $f_{3}$ degenerates to the particular case of $f_{1}$ with $\kappa=\pm1$.
    A graphical representation of this degeneration scheme is given in
    \Cref{fig:scheme}.
    Moreover, we observe that there are countably many values of $\kappa$ such
    that the associated evolution map, i.e. the map:
    \begin{equation}
        (J_{+}(t),J_{-}(t),J_{3}(t)) \mapsto (J_{+}(t+h),J_{-}(t+h),J_{3}(t+h))
        \label{eq:assmap3}
    \end{equation}
    from equation \eqref{eq:sl2evol} with $f_{1}=\kappa/2$ is periodic. 
    In \Cref{app:per} we show how to find these values, but we shall not discuss
    these degenerate cases further.
    \label{rem:degeneration}
\end{remark}

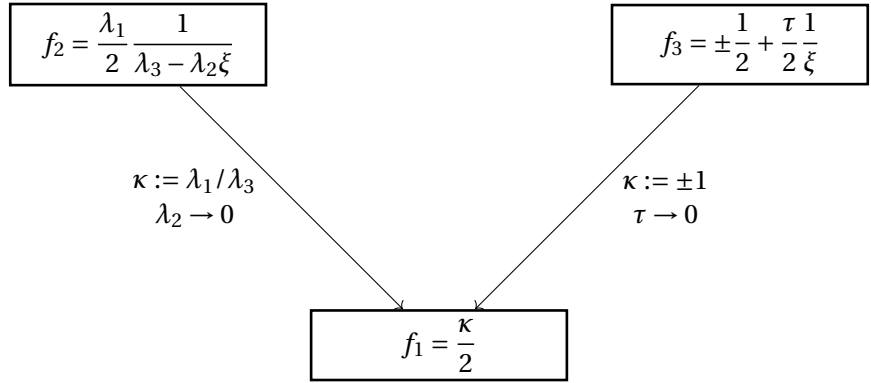
\begin{figure}[t!]
    \centering
    \begin{tikzpicture}
          \draw[line width=1pt]
          (-4,4)  node[entity,minimum  width =3.4cm] (f2) {$\displaystyle f_{2} = \frac{\lambda_{1}}{2} \frac{1}{\lambda_{3}-\lambda_{2}\xi}$}
            (4,4)  node[entity,minimum  width =3.4cm] (f3) {$\displaystyle f_{3} = \pm \frac{1}{2} + \frac{\tau}{2} \frac{1}{\xi}$}
            (0,0) node[entity,minimum  width =3.4cm] (f1) {$\displaystyle f_1=\frac{\kappa}{2}$};
            \path[<-] (f1) edge node[left=10pt,align=flush center] {$\kappa:=\lambda_1/\lambda_3$\\ $\lambda_2\to0$} (f2);
            \path[<-] (f1) edge node[right=10pt,align=flush center] {$\kappa:=\pm1$\\ $\tau\to0$} (f3);
\end{tikzpicture}
    \caption{Degeneration scheme for the functions $f_i$ in \eqref{eq:f123}.}
    \label{fig:scheme}
\end{figure}

\begin{proof}
    The proof of this result is through direct computation by considering
    the three cases $d=1,2,3$.  Indeed, the system already possesses
    one Casimir function, so if we are able to prove the existence of
    a functionally independent invariant, we have proven integrability.

\paragraph{Case $d=1$} 
This case is prototypical for the other two which can be proven
directly, so we will give all the details of the proof.
If $d=1$ the invariant has the following form:
\begin{equation}
    I_{1}(t) = a_{1,0,0} J_{+}(t) + a_{0,1,0} J_{-}(t) + a_{0,0,1} J_{3}(t).
    \label{eq:I1ini}
\end{equation}
Then, applying the translation by the step $h$ and using the
form of the system \eqref{eq:sl2evol} we obtain:
\begin{equation}
    \begin{aligned}
        I_{1}(t+h) &= a_{1,0,0} J_{-}(t) 
        + a_{0,1,0} \left( J_{+}(t) - 4 J_{3}(t) f\left( J_{-}(t) \right)
        +4 J_{-} f^{2}\left( J_{-}(t) \right) \right)
        \\
        &+ a_{0,0,1} \left(-J_{3}(t) + 2 J_{-}f\left( J_{-}(t) \right)\right).
    \end{aligned}
    \label{eq:I1init}
\end{equation}
We impose that $I_{1}(t+h)=I_{1}(t)$. This yields a polynomial 
expression of degree one in $J_{+}$ and $J_{3}$. Taking coefficients
with respect to these two variables, we obtain (after removing the non-zero
common factors):
\begin{equation}
    \begin{gathered}
        a_{0, 1, 0} - a_{1, 0, 0}=0,
        \quad
        2 f(J_{-}) a_{0, 1, 0}+ a_{0, 0, 1} =0,
        \\
        4 f^{2}(J_{-}) a_{0, 1, 0} + 2 f(J_{-}) a_{0, 0, 1} - a_{0, 1, 0} +a_{1, 0, 0} =0.
    \end{gathered}
    \label{eq:cc1}
\end{equation}
The first equation in \eqref{eq:cc1} implies $a_{1,0,0}=a_{0,1,0}$ 
with no other possibilities:
\begin{equation}
    2 f(J_{-}) a_{0, 1, 0}+ a_{0, 0, 1} =0,
    \quad
    f(J_{-})( 2 f(J_{-}) a_{0, 1, 0} +  a_{0, 0, 1} ) =0.
    \label{eq:cc1b}
\end{equation}
If $a_{0,1,0}\neq0$ the first equation in \eqref{eq:cc1b} implies 
that $f=\text{constant}$. This is not restricting because
$a_{0,1,0}=0$ yields $I_{1}\equiv 0$, which must be discarded.
Then, we write $f(\xi) = \kappa/2$, with $\kappa\neq0$, and 
the system reduces to the single algebraic equation:
\begin{equation}
    \kappa a_{0, 1, 0} + a_{0, 0, 1} =0.
    \label{eq:cc1c}
\end{equation}
Then the invariant \eqref{eq:I1ini} has the form:
\begin{equation}
    I_{1}(t) = a_{1,0,0} \left( J_{+}(t) +  J_{-}(t) - \kappa J_{3}(t)\right).
    \label{eq:I1fin}
\end{equation}
Scaling away the constant $a_{1,0,0}$ we proved formula \eqref{eq:I1},
while formula \eqref{eq:f1} was already established. The functional
independence of the set $\left\{ I_{1}, C \right\}$ is trivial, so the
system is clearly integrable. This concludes the case $d=1$.

\paragraph{Case $d=2$} 
If $d=2$ then the invariant has the following form:

\begin{equation}
    \begin{aligned}
        I_{2}(t) &= a_{2,0,0} J_{+}^2 + a_{0, 2, 0} J_{-}^{2}(t) + a_{0,0,2} J_{3}^{2}(t)
        \\
        &+ a_{1, 1, 0} J_{+}(t) J_{-}(t)+ a_{1, 0, 1} J_{+}(t) J_{3}(t)
        +a_{0,1,1} J_{-}(t) J_3(t)
        \\
        &+a_{1,0,0} J_{+}(t) + a_{0,1,0} J_{-}(t) + a_{0,0,1} J_{3}(t).
    \end{aligned}
    \label{eq:I2ini}
\end{equation}
Writing $I_{2}(t+h)=I_{2}(t)$ using the system \eqref{eq:sl2evol} and
taking the coefficients with respect to $J_{+}$ and $J_{3}$ we have the
following equations (divided by total degree $D$):
\begin{subequations}
    \begin{gather}
        D=2\colon
        \quad
        \begin{gathered}
            a_{0, 2, 0}-a_{2, 0, 0}=0, \quad
            f(J_{-})\left[4a_{0, 2, 0} f(J_{-}) +  a_{0, 1, 1}\right]=0, 
            \\
            8a_{0, 2, 0} f(J_{-}) +a_{0, 1, 1}+a_{1, 0, 1}=0,
        \end{gathered}
        \label{eq:cc2_1}
        \\
        D=1\colon
        \quad
        \begin{gathered}
            8a_{0, 2, 0} J_{-} f^{2}(J_{-}) +2 a_{0, 1, 1} J_{-} f(J_{-}) 
            +a_{0, 1, 0}-a_{1, 0, 0}=0,
            \\
            \begin{gathered}
                4J_{-}f^{2}(J_{-})\left(8a_{0, 2, 0}  f(J_{-}) +3a_{0, 1, 1} \right)
                \\
                +4\left[ \left(a_{0, 0, 2} + a_{1, 1, 0}\right) J_{-}
                +a_{0, 1, 0}\right] f(J_{-}) 
                \\
                +\left(a_{0, 1, 1} + a_{1, 0, 1}\right) J_{-}
                +2 a_{0, 0, 1}=0,
            \end{gathered}
        \end{gathered}
        \label{eq:cc2_2}
        \\
        D=0\colon
        \quad
        \begin{gathered}
            4 J_{-}f^{2}(J_{-}) 
            \left(4a_{0, 2, 0} f(J_{-})^2 +2 a_{0, 1, 1} f(J_{-}) +a_{0, 0, 2}  \right)
            \\
            +
            4\left( a_{1, 1, 0} J_{-} + a_{0, 1, 0}\right) 
            f^{2}(J_{-})
            +2\left( a_{1, 0, 1} J_{-} + a_{0, 0, 1}\right) f(J_{-}) 
            \\
            -\left(a_{0, 2, 0} - a_{2, 0, 0}\right)J_{-}
            -\left(a_{0, 1, 0} - a_{1, 0, 0}\right)=0.
        \end{gathered}
        \label{eq:cc2_3}
    \end{gather}
    \label{eq:cc2}%
\end{subequations}
We solve these equations starting from the highest degree:
$D=2$ \eqref{eq:cc2_1}.
To obtain a non-trivial solution not considered for $d=1$
we have to impose that all the coefficients with respect to
$f(J_{-})$ in \eqref{eq:cc2_1} vanish:
\begin{equation}
    a_{2, 0, 0}=a_{0, 2, 0} =a_{0, 1, 1} =a_{1, 0, 1}=0.
    \label{eq:cc2_1sol}
\end{equation}
This reduces our system to:
\begin{subequations}
    \begin{gather}
        D=1\colon
        \quad
        \begin{gathered}
            a_{0, 1, 0}-a_{1, 0, 0}=0,
            \\
            2\left[ \left(a_{0, 0, 2} + a_{1, 1, 0}\right) J_{-}
            +a_{0, 1, 0}\right] f(J_{-}) + a_{0, 0, 1}=0,
        \end{gathered}
        \label{eq:cc2_2b}
        \\
        D=0\colon
        \quad
        \begin{gathered}
            4\left[ \left(a_{0, 0, 2} +  a_{1, 1, 0} \right) J_{-}
            +  a_{0, 1, 0}\right]f^{2}(J_{-})
            \\
            +2a_{0, 0, 1} f(J_{-}) -a_{0, 1, 0}+a_{1, 0, 0}
            =0.
        \end{gathered}
        \label{eq:cc2_3b}
    \end{gather}
    \label{eq:cc2b}%
\end{subequations}
Solving the first equation for $D=1$, i.e. putting $a_{0,1,0}=a_{1,0,0}$,
and discarding the trivial solutions, we have that all the equations
reduce to:
\begin{equation}
    2\left[ \left(a_{0, 0, 2} + a_{1, 1, 0}\right) J_{-}
    +a_{0, 1, 0}\right] f(J_{-}) + a_{0, 0, 1}=0.
    \label{eq:cc2c}
\end{equation}
That is:
\begin{equation} 
     f(J_{-}) =-\frac{1}{2}\frac{a_{0, 0, 1}}{ \left(a_{0, 0, 2} + a_{1, 1, 0}\right) J_{-}
    +a_{0, 1, 0}}.
    \label{eq:f2orig}
\end{equation}
Of these last four parameters we can choose three independent ones as follows:
\begin{equation}
    a_{0,0,1} = \lambda_{1},
    \quad
    a_{0,0,2} = -\lambda_{2} - a_{1,1,0},
    \quad 
    a_{0,1,0}=\lambda_{3}.
    \label{eq:parfix2}
\end{equation}
This brings \eqref{eq:f2orig} into the form of \eqref{eq:f2},
and the invariant \eqref{eq:I2ini} into the following form:
 \begin{equation} 
    I_{2}(t) = 
    \lambda_{3}\left[J_{+}(t) + J_{-}(t)\right] - \lambda_{1} J_{3}(t) - \lambda_{2} J_{3}^{2}(t)
    +a_{1,1,0} \left[ J_{+}(t)J_{-}(t)-J_{3}^{2}(t) \right].
    \label{eq:I2orig}
\end{equation}
Now, note that the coefficient of the $a_{1,1,0}$ is exactly the
Casimir function $C$ of $\Sl_{2}(\R)$, see formula \eqref{eq:sl2cas}.
From \Cref{thm:rad} this invariant is admitted by the system
\eqref{eq:sl2evol} \emph{regardless of the function $f=f(\xi)$}. So, this
invariant is trivial and we can safely discard it by putting $a_{1,1,0}=0$
in \eqref{eq:I2orig} (we are already taking into account its existence).
This proves formula \eqref{eq:I2}.  The functional independence of the
set $\left\{ I_{2}, C \right\}$ is trivial, so the system is clearly
integrable. This concludes the case $d=2$.

\paragraph{Case $d=3$} 
Let us consider now the case $I_3(t)$ in formula \eqref{eq:Idgen} (we
omit the explicit form because it is rather cumbersome).
We employ the same strategy we employed before: from the identity $I_{3}(t+h)=I_{3}(t)$ equate all the coefficients in $J_{+}$, $J_{3}$ to zero starting from the higher degree.
Taking the coefficients of degree 3 in $J_{+}$, $J_{3}$ in $I_{3}(t+h)=I_{3}(t)$ 
we have the following equations:
\begin{equation}
    \begin{aligned}
        J_{+}^{3} & \colon a_{{0,3,0}}-a_{{3,0,0}}=0,
        \\
        J_{+}^{2}J_{3} &\colon
        -12 f \left( J_{-} \right) a_{{0,3,0}}-a_{{0,2,1}}-a_{{2,0,1}}=0,
        \\
        J_{+}J_{3}^{2} &\colon
        48  f^{2} \left( J_{-} \right) a_{{0,3,0}}
        +8 f \left( {\it J_{-}} \right) a_{{0,2,1}}+ a_{{0,1,2}}-a_{{1,0,2}}=0,
        \\
        J_{3}^{3} &\colon
        -32  f^{3} \left( J_{-} \right) a_{{0,3,0}}
        -8  f^{2} \left( J_{-} \right) a_{{0,2,1}} 
        -2 f \left( J_{-} \right) a_{{0,1,2}}-  a_{{0,0,3}}=0.
    \end{aligned}
    \label{eq:cc3D3}
\end{equation}
Since the only possible solution for $f$ from this set of equations is
the constant, we can safely put all the coefficients of $f$ to zero
and obtain:
\begin{equation}
    a_{{0,0,3}}=a_{{0,1,2}}=a_{{0,2,1}}=a_{{0,3,0}}=a_{{1,0,2}}
    =a_{{2,0,1}}=a_{{3,0,0,0}}=0.
    \label{eq:cc3D3sol}
\end{equation}
Inserting these values into $I_{3}(t+h)=I_{3}(t)$ the degree three
terms disappear, and we can consider the degree two ones:
\begin{equation}
    \begin{aligned}
        J_{+}^{2} & \colon 
        \left(a_{{1,2,0}}-a_{{2,1,0}}\right) J_{-} +a_{{0,2,0}}-a_{{2,0,0}}=0,
        \\
        J_{+}J_{3} & \colon 
        8\left( J_{-} a_{{1,2,0}}- a_{{0,2,0}}\right) f \left( J_{-} \right)
            +2 J_{-} a_{{1,1,1}}+a_{{0,1,1}}+a_{{1,0,1}}=0,
        \\
        J_{3}^{2} & \colon 
        f \left( J_{-} \right) \left[
            4 \left(J_{-}a_{{1,2,0}} +  a_{{0,2,0}}\right) f \left( {\it J_{-}}\right)
            + J_{-}  a_{{1,1,1}}+  a_{{0,1,1}}\right]=0.
    \end{aligned}
    \label{eq:cc3D2}
\end{equation}
Taking the coefficients with respect to $J_{-}$ the first equation 
in \eqref{eq:cc3D2} yields:
\begin{equation}
    a_{{2,1,0}}=a_{{1,2,0}}, \quad a_{{2,0,0}}=a_{{0,2,0}}.
    \label{eq:cc3D2sol}
\end{equation}
Discarding the trivial solution $f\left( J_{-} \right)=0$ we have that
$f(J_{-})$ from the second and the third equation in \eqref{eq:cc3D2}
must be compatible. 
The compatibility condition is:
\begin{equation}
    a_{{1,0,1}}=a_{{0,1,1}}.
    \label{eq:cc3D2solbis}
\end{equation}
This yields a unique value for $f(J_{-})$:
\begin{equation}
    f \left( J_{-} \right) =-\frac{1}{4} 
    \frac { J_{-} a_{{1,1,1}}+a_{{0,1,1}}}{ J_{-} a_{{1,2,0}}+a_{{0,2,0}}}.
    \label{eq:f3sol0}
\end{equation}

At this point we can insert this value into $I_{3}(t+h)=I_{3}(t)$
which now consists only of the linear and the degree zero terms.
Now we can clear denominators and take coefficients with respect to
$J_{-}$. This leaves us with a system of $26$ algebraic equations we
can solve with a CAS, e.g. \texttt{Maple}. We omit the explicit form of
the system, but we comment that we obtain 12 different solutions with
the \texttt{solve} command from \texttt{Maple 2016}. Upon inspection
only two new solutions do not produce trivial solutions and satisfy
the condition that $a_{1,2,0}a_{0,2,0}\neq0$. The final form of the
coefficients we obtained is shown in Appendix \ref{app:expl}.

This yields the following form for the function $f$:
\begin{equation}
    f\left( J_{-} \right) = \pm \frac{1}{2} - \frac{a_{0, 1, 1}}{4a_{1, 2, 0}} \frac{1}{J_{-}}
    \label{eq:f3sol1}
\end{equation}
Since $a_{1,2,0}\neq0$ we perform the scaling 
$a_{0,1,1}=-2a_{1,2,0}\tau$. This proves formula \eqref{eq:f3}.
To conclude the proof we notice that using the results as displayed
in Appendix \ref{app:expl} and the scaling just presented we have:
\begin{equation}
    \begin{aligned}
        I_{3} &=
        a_{1,2,0}\left(J_{+}(t)+J_{-}(t)\mp2 J_3(t)\right) 
        \left[\tau^2-2 \tau J_{3}(t)+J_{+}(t) J_{-}(t)\right]
        \\
        &+ a_{1,1,0} \left[ J_{+}(t)J_{-}(t)-J_{3}^{2}(t) \right].
    \end{aligned}
    \label{eq:I3sol0}
\end{equation}
Now, note that the coefficient of the $a_{1,1,0}$ is exactly the
Casimir function $C$ of $\Sl_{2}(\R)$, see formula \eqref{eq:sl2cas}.
From \Cref{thm:rad} this invariant is admitted by the system
\eqref{eq:sl2evol} \emph{regardless of the function $f=f(\xi)$}. So,
this invariant is trivial and we can safely discard it by putting
$a_{1,1,0}=0$ (we are already taking into account its existence).
Then, analogously to the case $d=1$ we can scale away the
constant $a_{1,2,0}$ and we arrive at the expression \eqref{eq:I3}.
The functional independence of the set $\left\{ I_{3}, C \right\}$ is
trivial, so the system is clearly integrable. This concludes the case
$d=3$ and hence concludes the proof of the theorem.
\end{proof}

For invariants of degree $d>3$, by direct computation it is possible to
prove the following result:

\begin{lemma}
    For all $d>3$ the translated polynomial \eqref{eq:Idgen} on the 
    solution of the system \eqref{eq:sl2evol} is:
    \begin{equation}
        I_{d}(t+h)
        = \sum_{1\leq i+j+k \leq d} 
        \;\;
        \sum_{p+q+r=j}\;\;
        \sum_{s=0}^{k}
        A_{i,j,k;p,q,r;s}
        J_{+}^{p} J_{-}^{i+k-s}  
        f^{q+2 r+k-s}(J_{-})
        J_{3}^{q+s},
        \label{eq:Idtrasl}
    \end{equation}
    where
    \begin{equation}
        A_{i,j,k;p,q,r;s} = 
        \left( -1 \right)^{q+s}
        a_{i,j,k} \binom{j}{p,q,r}\binom{k}{s}
        2^{2 q + 2 r + k-s}.
        \label{eq:Aijk}
    \end{equation}
    \label{lem:Idtrasl}
\end{lemma}

Then, with a technique analogue to the one that we used to prove Theorem
\ref{thm:invariant} we can see that for $d=4,5$ no different integrable
systems are found. The calculations are of increasing complexity, so
we don't show them here. However, this leads to the following conjecture:

\begin{conjecture}
    If $d>3$ and the function $f$ is different from the ones in the formula
    \eqref{eq:f123}, then only trivial invariants are possible. Here,
    by a trivial invariant, we mean a linear combination of the power
    of the Casimir function \eqref{eq:sl2cas}:
    \begin{equation}
        I_\text{triv}(t) = 
        \sum_{k=1}^{\lfloor d/2\rfloor} a_{k} C^{k}.
        \label{eq:Itrivial}
    \end{equation}
\end{conjecture}

\section{Study of the obtained systems}
\label{sec:sys}

In this section, we present a more detailed study of the realisation 
in canonically conjugated coordinates $(\vec{q}(t),\vec{p}(t))$ of the
integrable systems we found in the previous section.

\subsection{Degree one invariant}

We consider the system \eqref{eq:f1} in the realisation of
$\Sl_{2}(\R)$ \eqref{eq:canrel} with canonical coordinates $(
\vec{q}(t),\vec{p}(t) )$.
The equation assumes the following form:
\begin{equation}
    q_{k}(t+h)
    \sqrt{1-\frac{b_{k}}{q_{k}^{2}(t)q_{k}^{2}(t+h)}}
    +
    q_{k}(t-h)
    \sqrt{1-\frac{b_{k}}{q_{k}^{2}(t)q_{k}^{2}(t-h)}}
    =
    \kappa q_{k}.
    \label{eq:dELsl2lin}
\end{equation}
From formula \eqref{eq:pkgen} the symplectic form of the system is:
\begin{subequations}
    \begin{gather}
        q_{k}(t+h)
        \sqrt{1-\frac{b_{k}}{q_{k}^{2}(t)q_{k}^{2}(t+h)}}
        +
        p_{k}(t)
        =
        \kappa q_{k},
        \label{eq:eqHamqsl2lin}
        \\
        p_{k}(t+h) =
        q_{k}(t)\sqrt{1-\frac{b_{k}}{q_{k}^{2}(t)q_{k}^{2}(t+h)}}.
        \label{eq:eqHampsl2lin}
    \end{gather}
    \label{eq:eqHamsl2lin}%
\end{subequations}
The case when $b_{k}=0$ for all $k=1,\dots,N$ is a linear system
and it was considered in \cite[Example 1]{GLT_coalgebra}, so we will 
consider only the case when $b_{k}\neq0$.

\begin{remark}
    We remark that the system \eqref{eq:dELsl2lin} is symmetric with respect to the
    transformation:
    \begin{equation}
        \vec{\widehat{q}}^{(i)}(t) = \left(q_{1}(t),\dots,-q_{i}(t),\dots,q_{N}(t)  \right),
        \label{eq:dsymm}
    \end{equation}
    or its symplectic form \eqref{eq:eqHamsl2lin} is symmetric under the
    transformation:
    \begin{subequations}
        \begin{align}
            \vec{\widehat{q}}^{(i)}(t) &= \left(q_{1}(t),\dots,-q_{i}(t),\dots,q_{N}(t)  \right),
            \\
            \vec{\widehat{p}}^{(i)}(t) &= \left(p_{1}(t),\dots,-p_{i}(t),\dots,p_{N}(t)  \right).
        \end{align}
        \label{eq:dsymmcan}%
    \end{subequations}
    This implies that the system can be defined outside its natural
    domain as described \Cref{thm:rad} by reflection. Reflecting with
    respect to each coordinate axis we obtain that the symplectic map
    \eqref{eq:eqHamsl2lin} can be defined on the set
    \begin{equation}
        U=\left(\prod_{k=1}^{N}\left\{ x_{k}\neq0 \right\}\right)^{\times2}.
        \label{eq:setU}
    \end{equation}
    \label{rem:defSW}
\end{remark}

The coalgebraic invariant \eqref{eq:I1} has the following explicit 
form\footnote{We added a cosmetic $1/2$ factor to better compare with the
continuous case.}:
\begin{equation}
    H_{1} = \frac{1}{2}\sum_{k=1}^{N} \left[ p_{k}^{2}(t) +\frac{b_{k}}{q_{k}^{2}(t)} 
    +q_{k}^{2}(t)-\kappa q_k(t)p_k(t)  \right].
    \label{eq:H0lin}
\end{equation}
From the coalgebra construction, we get the following two sets of invariants
\begin{equation}
    \mathcal{L}_{1} = \left\{ H_{1},\sC^{[2]},\dots,\sC^{[N]} \right\},
    \quad
    \mathcal{R}_{1} = \left\{ H_{1},\sC_{[2]},\dots,\sC_{[N]} \right\},
    \label{eq:invsetslin}
\end{equation}
By induction on the d.o.f. $N$, it is easy to see that both
sets are functionally independent. As expected, this implies that the
system \eqref{eq:dELsl2lin} is Liouville integrable. Similarly the set 
\begin{equation}
    \mathcal{I}_{1} = \left\{ H_{1},\sC^{[2]},\dots,\sC^{[N]},
    \sC_{[2]},\dots,\sC_{[N-1]} \right\},
    \label{eq:fullinvsetslin}
\end{equation}
is made of functionally independent functions which implies the system
\eqref{eq:dELsl2lin} is QMS.
As remarked in \cite{GLT_coalgebra} this is the best we expect for
general discrete-time systems with $\Sl_{2}(\R)$ coalgebra symmetry.
However, this case is special because it is possible to consider a different set of invariants constructed from the following elements:
\begin{equation}
   \mathsf{H}_{[k]}:= \overbrace{1 \otimes \dots \otimes 1}^{k-1} \otimes \mathsf{H} \otimes \overbrace{1 \otimes \dots \otimes 1}^{N-k}, \quad   k = 1,\dots,N,
    \label{eq:linconstr}
\end{equation}
where:
\begin{equation}
    \mathsf{H} = \frac{1}{2} (J_++J_--\kappa J_3) \, .
    \label{eq:Hk}
\end{equation}
Those $N$ elements, at fixed realisation, result in the following functions:
\begin{equation}
    \mathsf{H}_{[k]} = \frac{1}{2} 
    \left( p_{k}^{2}  + q_{k}^{2}+\frac{b_{k}}{q_{k}^{2}} - \kappa q_{k}p_{k}\right),
    \quad
    k = 1,\dots,N.
    \label{eq:Hkrel}
\end{equation}
This gives us another set of commuting invariants:
\begin{equation}
    \mathcal{C}_{1} = \left\{ \mathsf{H}_{[1]},\dots,\mathsf{H}_{[N]} \right\},
    \label{eq:invsetslin2}
\end{equation}
which proves Liouville's integrability again.
Note that:
\begin{equation}
    H_{1} = \sum_{k=1}^{N} \mathsf{H}_{[k]}.
    \label{eq:H1sum}
\end{equation}
Then, the following set of invariants:
\begin{equation}
    \mathcal{S}_{1} = \left\{\mathsf{H}_{[1]},\dots,\mathsf{H}_{[N]},
    \sC^{[2]},\dots,\sC^{[N]}  \right\},
    \label{eq:mslin}
\end{equation}
is a set of $2N-1$ functionally independent invariants.
For explicit commutation relations involving left and right
Casimir invariants, also with these additional constants, we refer
the reader to \cite{Latini_2019, LATINI2021168397, Latini_2021}.
Functional independence can be proven as follows. Define the point
$\vec{m}=(1,\dots,1,0,\dots,0)$. Then from direct computation we have:
\begin{subequations}
    \begin{align}
        \vec{v}_{k} &= \left.\grad H_{k}\right|_{\vec{m}} =
            \big( \underbrace{0,\dots,0,\overbrace{-b_{m}+1}^{\text{$m$th element}},0,\dots,0}_{N},
            \underbrace{0,\dots,0,\overbrace{-\kappa}^{\text{$m$th element}},0,\dots,0}_{N} \big),
        \label{eq:vk}
        \\
        \vec{w}_{m} &= \left.\grad \sC^{[m]}\right|_{\vec{m}} 
            =\big(\underbrace{b_{0}-mb_{1},\dots,b_{0}-mb_{m}}_{m},0,\dots,0\big),
        \label{eq:wm}
    \end{align}
    \label{eq:vkwm}%
\end{subequations}
for $k=1,\dots,N$, and $m=2,\dots,N$ respectively and we 
defined $b_{0} = \sum_{i=1}^{m}b_{i}$.
Then we can define:
\begin{equation}
    \begin{aligned}
        \vec{u}_{m} &= \vec{w}_{m} - \sum_{i=1}^{m} \frac{b_{0}-mb_{i}}{-b_{i}+1} \vec{v}_{i}
        \\
        &= \kappa
        \Biggl( \underbrace{0,\dots,0}_{N},
        \underbrace{\frac{b_{0}-mb_{1}}{-b_{1}+1},\dots,\frac{b_{0}-mb_{m}}{-b_{m}+1}}_{m},
        0,\dots,0 \Biggr),
    \end{aligned}
    \quad m=2,\dots,N.
    \label{eq:um}
\end{equation}
So, we have that the Jacobian in $\vec{m}$
\begin{equation}
    J_{\vec{m}} = 
    \left.
    \frac{\partial\left(\mathsf{H}_{1},\dots,\mathsf{H}_{N}, \sC^{[2]},\dots,\sC^{[N]} \right)}{\partial\left( \vec{q},\vec{p} \right)}
    \right|_{\vec{m}} =
    \begin{pmatrix}
        \vec{v}_{1}
        \\
        \vdots
        \\
        \vec{v}_{N}
        \\
        \vec{w}_{2}
        \\
        \vdots
        \\
        \vec{w}_{N}
    \end{pmatrix},
    \label{eq:J0def}
\end{equation}
is equivalent through Gauss elimination to the upper triangular
matrix:
\begin{equation}
    \widetilde{J}_{\vec{m}} = 
    \begin{pmatrix}
        \vec{v}_{1}
        \\
        \vdots
        \\
        \vec{v}_{N}
        \\
        \vec{u}_{2}
        \\
        \vdots
        \\
        \vec{u}_{N}
    \end{pmatrix}.
    \label{eq:J0tilde}
\end{equation}
Since all the pivots of $\widetilde{J}_{\vec{m}}$ are non-zero it follows
that $\rank J_{\vec{m}}=2N-1$ proving functional independence in
the open set
\begin{equation}
    U_{(+,\dots,+)} = O_{(+,\dots,+)}\times \R^{N},
    \label{eq:Uplus}
\end{equation}
where $O_{(+,\dots,+)}$ is the open positive orthant of $\R^{N}$ \cite{Roman2007book}.
A similar argument runs in all other possible open sets 
\begin{equation}
    U_{(\pm,\dots,\pm)} = O_{(\pm,\dots,\pm)}\times \R^{N},
    \quad
    O_{(\pm,\dots,\pm)}=\left\{ \pm q_{1}>0,\dots,\pm q_{N}>0 \right\},
    \label{eq:Ugen}
\end{equation}
evaluating the Jacobian on the points $(\pm1,\dots,\pm1,0,\dots,0)\in
U_{(\pm,\dots,\pm)}$. Noting that solutions in one orthant evolve
inside the same orthant, while the coordinate lines $\left\{ q_{k}=0
\right\}_{k=1,\dots,N}$ are not accessible by the evolution we have that
these sets exhaust all the points of the phase space, 
because:
\begin{equation}
    U \subset \bigcup_{i_{k}\in \left\{ \pm \right\}} U_{(i_{1},\dots,i_{N})},
    \label{eq:cups}
\end{equation}
where the set $U$ was defined in formula \eqref{eq:setU}.
So, this proves
the functional independence of the invariants everywhere they are defined.
This allows us to conclude that the dSW system \eqref{eq:dELsl2lin} is MS.

In Figure \ref{fig:dSW_dim3} we show an orbit of \eqref{eq:dELsl2lin}
near the fixed point
\begin{equation}
    \vec{q}_{0} = 
    \left( \left(\frac{b_{1}}{1-(\kappa/2)^{2}}\right)^{1/4},
        \dots,
    \left(\frac{b_{N}}{1-(\kappa/2)^{2}}\right)^{1/4} \right).
    \label{eq:fixed}
\end{equation}
Note that the fixed points of the map \eqref{eq:dELsl2lin} are
all of the following form:
\begin{equation}
    \vec{q}_{0}^{(\pm,\dots,\pm)} = 
    \left( \pm\left(\frac{b_{1}}{1-(\kappa/2)^{2}}\right)^{1/4},
        \dots,
    \pm\left(\frac{b_{N}}{1-(\kappa/2)^{2}}\right)^{1/4} \right).
    \label{eq:fixedall}
\end{equation}

\begin{figure}[htp]
    \centering
    \includegraphics{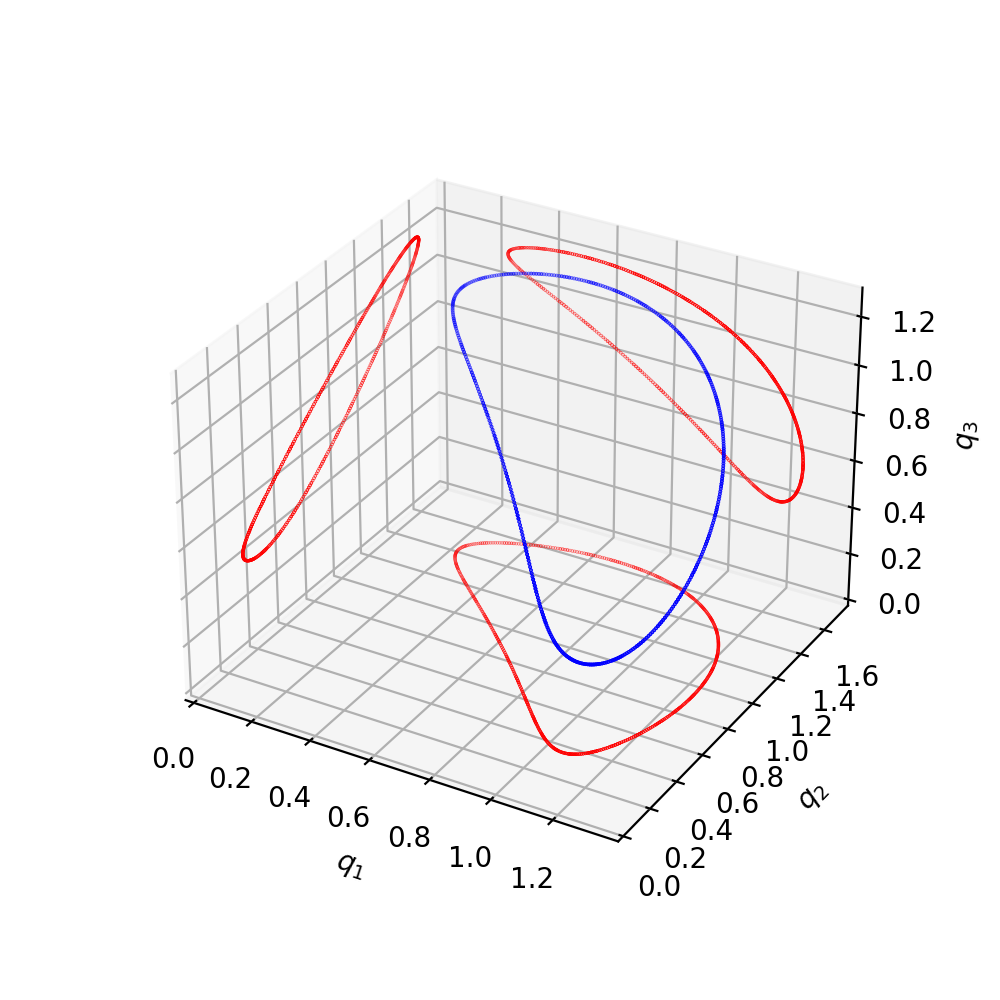}
    \caption{An orbit of the dSW system \eqref{eq:dELsl2lin} in three
    d.o.f. for an initial condition near the fixed point \eqref{eq:fixed}.
    The value of the parameters are $b_{1}=\num{0.01}$, $b_{2}=\num{0.03}$,
    $b_{3}=\num{0.02}$, and $\kappa = 4\sqrt{2}/3$ with $M=\num{10 000}$ iterations.
    The projections of the orbit on the coordinate planes are shown.}
    \label{fig:dSW_dim3}
\end{figure}

Next, we note that the system \eqref{eq:dELsl2lin} is a discretisation
of the Smorodinsky--Winternitz (SW) oscillator (caged isotropic harmonic oscillator) 
\cite{Fris1965,Evans1990}.
Indeed, under the following scaling:
\begin{equation}
    q_{k}(t) = h Q_{k}(t),
    \quad 
    b_{k} = h^{6}\beta_{k},
    \quad
    \kappa = 2 - h^{2} \omega^{2},
    \label{eq:scal}
\end{equation}
in equation \eqref{eq:dELsl2lin} we obtain that the leading order
as $h\to0$ is SW oscillator in Lagrangian form:
\begin{equation}
    h^{3}\left( \ddot{Q}_{k}+\omega^{2} Q_{k} - \frac{\beta_{k}}{Q_{k}^{3}} \right)+
    \bigO(h^{4}) = 0.
    \label{eq:SW}
\end{equation}
In the same way, writing the invariant \eqref{eq:eqHamqsl2lin}
using the definition of the canonical momentum \eqref{eq:eqHampsl2}
and $\dot{Q}_{k}=P_{k}$ we obtain:
\begin{equation}
    H_{1} =
    h^{4} H_\text{SW} 
    +\bigO( h^{5} ),
    \quad
    H_\text{SW} = 
    \sum_{k=1}^{N} \left( P_{k}^{2} + \omega^{2} Q_{k}^{2} + \frac{\beta_{k}}{Q_{k}^{2}} \right).
    \label{eq:HlinSW}
\end{equation}
That is, at order 4 the Hamiltonian of the SW oscillator appears. The SW
system is a well-known MS system whose invariants can be produced through
the coalgebra symmetry method \cite{Ballesteros_et_al2008PhysAtomNuclei}.
In particular, we note that the other invariants \eqref{eq:Cm} and
\eqref{eq:Hk} are preserved in form through the scaling \eqref{eq:scal}.
Indeed, noting that in the limit $h\to0$:
\begin{equation}
    L_{i,j} = \ell_{i,j} h^{3} + \bigO(h^{3}),
    \quad 
    \ell_{i,j} = P_{j}Q_{i}-P_{i}Q_{j}.
    \label{eq:Lijlim}
\end{equation}
we have that the invariants \eqref{eq:Cm} are self-similar
at order 6 in $h$:
\begin{equation}
    \sC^{[m]}(\vec{q},\vec{p}) = h^{6}  \sC^{[m]}(\vec{Q},\vec{P}) + \bigO(h^{7}),
    \quad
    \sC_{[m]}(\vec{q},\vec{p}) = h^{6}  \sC_{[m]}(\vec{Q},\vec{P}) + \bigO(h^{7}),
    \label{eq:CmHk}
\end{equation}
while:
\begin{equation}
    \mathsf{H}_{k} =
    h^{4} H_{[k]}
    +\bigO( h^{5} ),
    \quad
    H_{[k]} = 
    \frac{1}{2}\left(P_{k}^{2} + \omega^{2} Q_{k}^{2} + \frac{\beta_{k}}{Q_{k}^{2}}\right).
    \label{eq:HlinSWk}
\end{equation}

A natural generalisation of the SW system arises if we consider anisotropy. That is, considering the following dynamical system:
\begin{equation}
    H(\boldsymbol{\omega}) = 
    \frac{1}{2}\sum_{k=1}^{N} \left( P_{k}^{2} + \omega^{2}_{k} Q_{k}^{2} + \frac{\beta_{k}}{Q_{k}^{2}} \right).
    \label{eq:Hlinanis}
\end{equation}
This system is clearly integrable considering the separation
of variables in Cartesian coordinates.
Following \cite{Evans2008}, see also \cite{GubLatDrach,Rodriguez2009}, 
the case when $\boldsymbol{\omega} =  \omega \left( l_{1},\dots,l_{N} \right)$
where $(l_{1},\dots,l_{N})\in\Z^{N}$ and the integers $l_{i}$ are coprime
is MS. In the discrete-time setting we consider the Lagrangian:
\begin{equation}
    L(\vec{c}) = 
    \sum_{k=1}^{N} 
    \left[
    \int^{q_{k}(t)q_{k}(t+h)}\sqrt{1-\frac{b_{k}}{\xi^{2}}}\ud\xi
    -\frac{c_{k}}{2}q_{k}^{2}(t)\right].
    \label{eq:dLanis}
\end{equation}
Then, using the analogue invariants \eqref{eq:Hk} 
\begin{equation}
    \mathsf{H}_{k}(c_{k}) = \frac{1}{2}\left(p_{k}^{2}  + q_{k}^{2}+\frac{b_{k}}{q_{k}^{2}} - c_{k} q_{k}p_{k}\right),
    \quad
    k = 1,\dots,N,
    \label{eq:Hkck}
\end{equation}
we prove that the dEL equations of the Lagrangian \eqref{eq:dLanis} 
are integrable. Finally, under the scaling, the associated dEL equations:
\begin{equation}
    q_{k}(t) = h Q_{k}(t),
    \quad 
    b_{k} = h^{6}\beta_{k},
    \quad
    c_{k} = 2 - h^{2} \omega_{k}^{2},
    \label{eq:scalbis}
\end{equation}
reduce to the Hamiltonian system associated to \eqref{eq:Hlinanis}.

Differently from the continuum case, the experimental evidence seems to
suggest that the system \eqref{eq:dLanis} is not MS. Indeed, from an
experimental study of the orbits of the system \eqref{eq:dLanis} it is
evident that they do not lie on a closed curve, but rather they are a
space-filling curve, dense in the phase space. For instance, Figure
\ref{fig:anis} shows an orbit of the system \eqref{eq:dLanis}
in two d.o.f. in the scaling \eqref{eq:scalbis} with parameters such
that $\omega_{1}/\omega_{2} = 2/3\in\Q$. However, despite the frequency
ratio is rational after $M=\num{1 000 000}$ iterations a rectangle in the phase
space is almost completely filled.

The orbit is taken in a neighbourhood of the fixed point:
\begin{equation}
    \vec{q}_{0}(\omega_{1},\omega_{2})=
    \left(
        \sqrt{\frac{2}{\omega_{1}}}
        \left(\frac{\beta_1}{4-h^2\omega_{1}^2}\right)^{\frac{1}{4}},
        \sqrt{\frac{2}{\omega_{2}}}
        \left(\frac{\beta_2}{4-h^2\omega_{2}^2}\right)^{\frac{1}{4}}
    \right).
    \label{eq:fixedck}
\end{equation}

\begin{figure}[hbt]
\centering
\subfloat[][$M=\num{1000}$ iterations.]{%
    \includegraphics[width=0.5\textwidth]{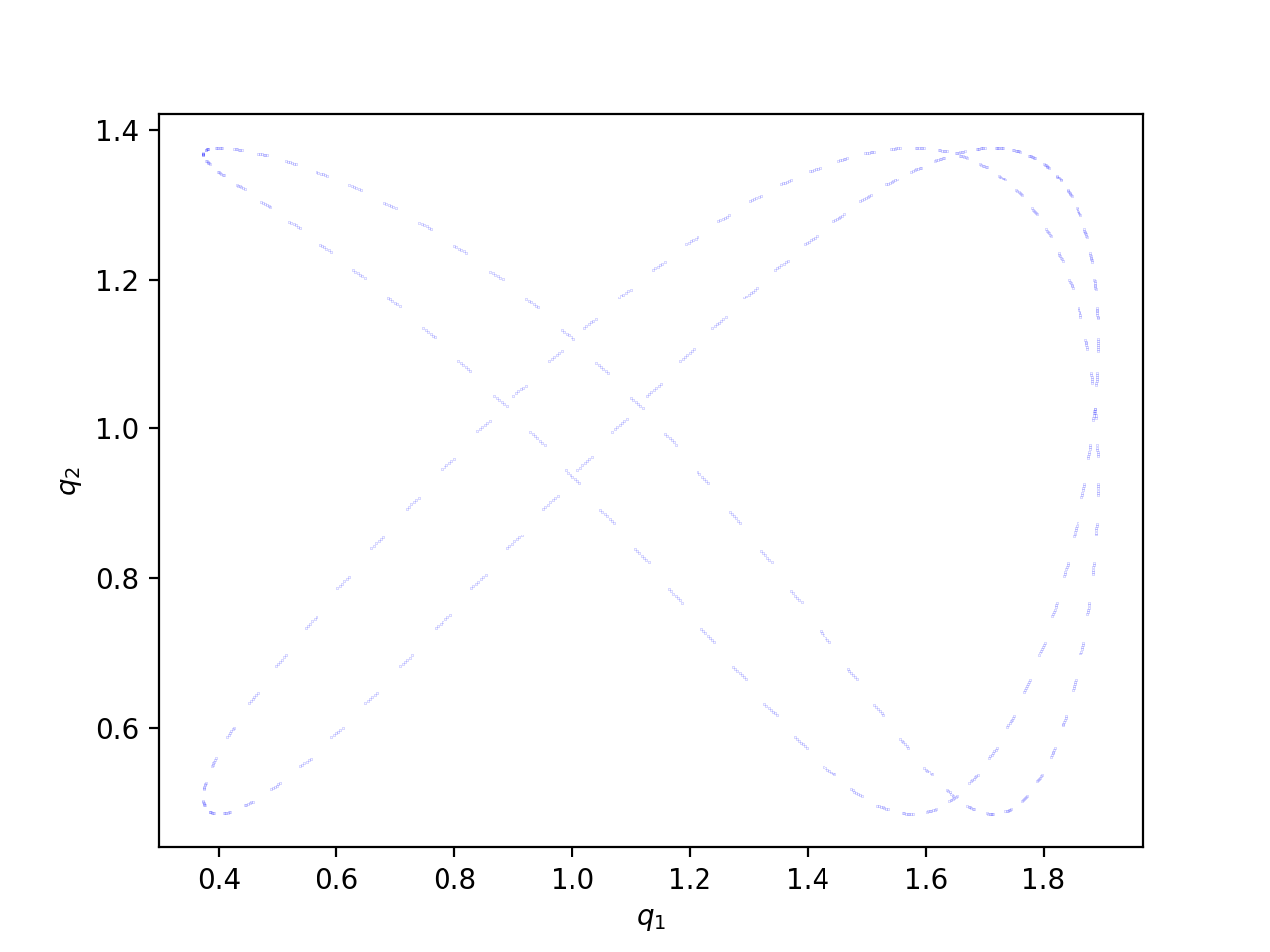}
    }
    \subfloat[][$M=\num{10000}$ iterations.]{
    \includegraphics[width=0.5\textwidth]{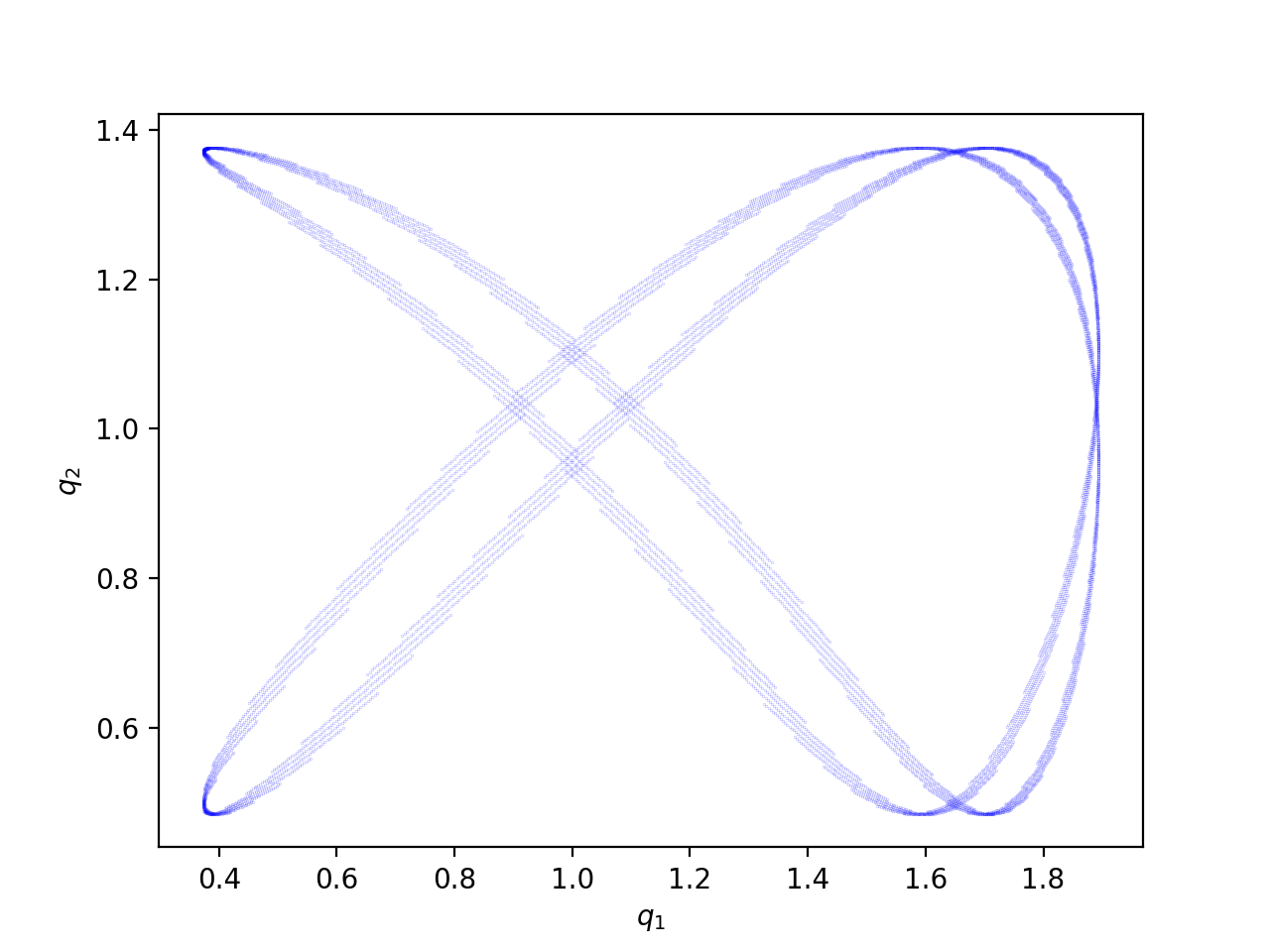}
}
\\
\subfloat[][$M=\num{100000}$ iterations.]{
    \includegraphics[width=0.5\textwidth]{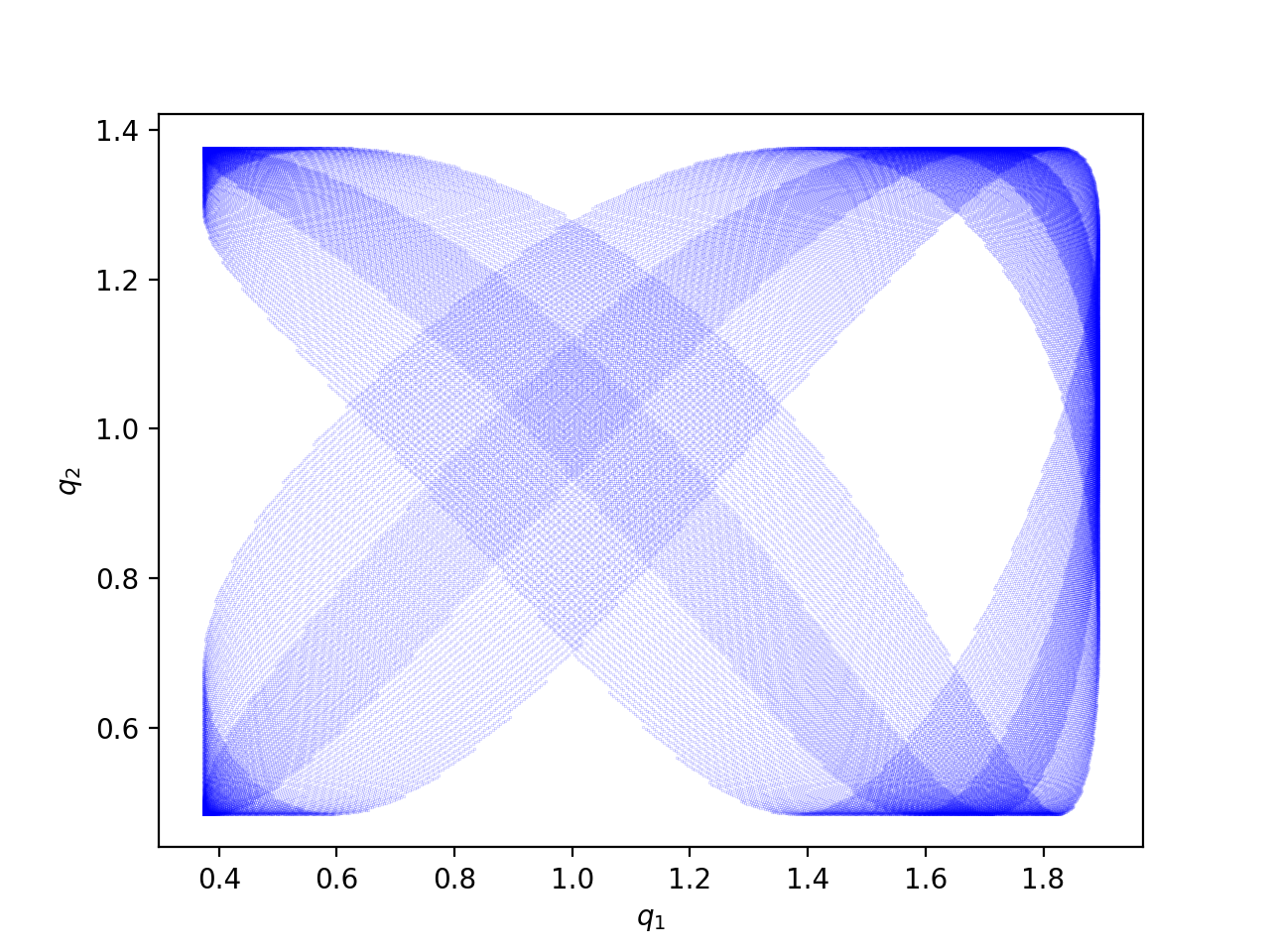}
}
\subfloat[][$M=\num{1000000}$ iterations.]{
    \includegraphics[width=0.5\textwidth]{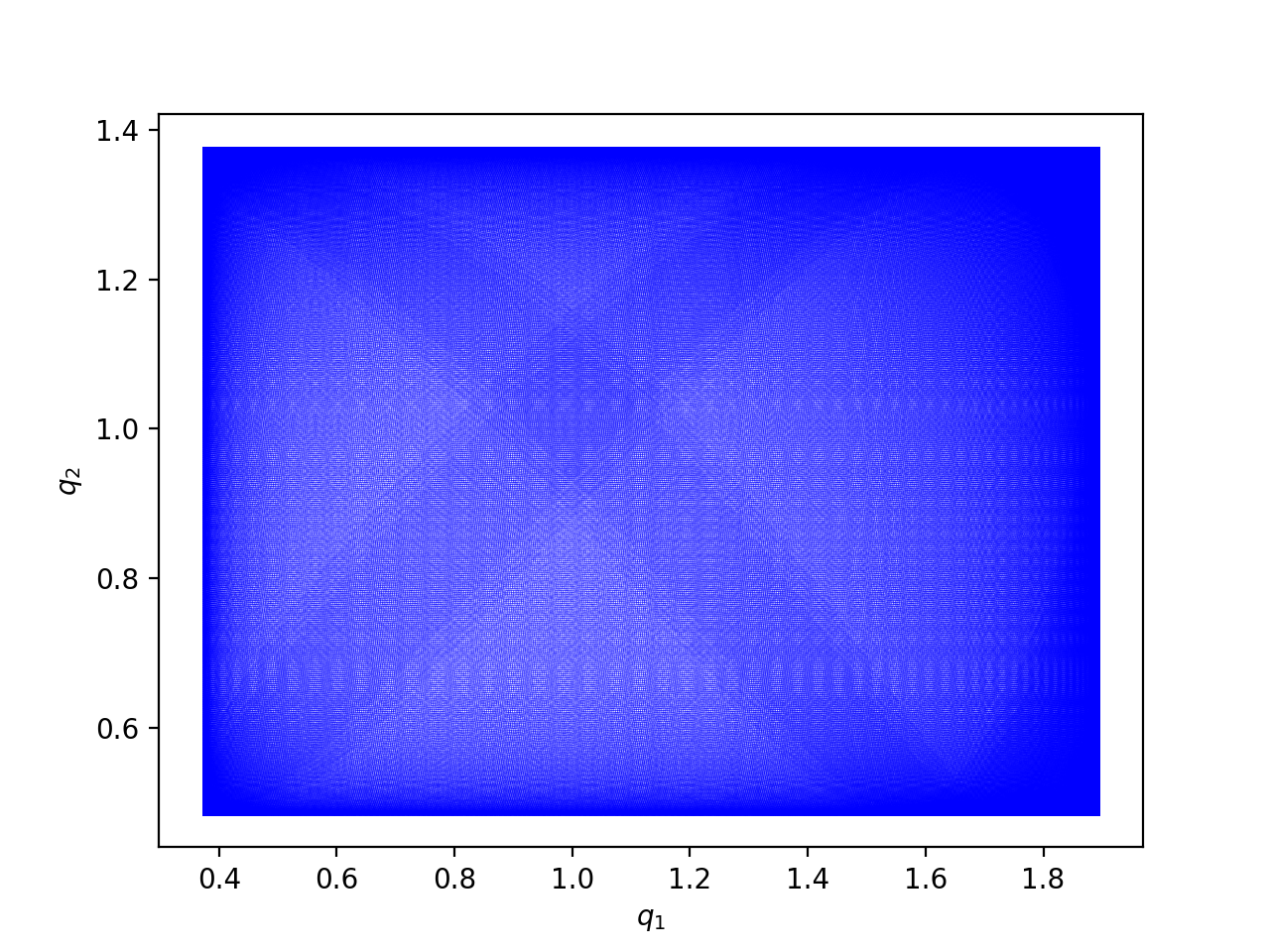}
}
\caption{An orbit of the two d.o.f. anisotropic dSW oscillator
    \eqref{eq:dLanis} scale according to \eqref{eq:scalbis} with
    values $h=1/50$, $\beta_{1} = 2$, $\beta_{2} = 4$, $\omega_{1} = 2$,
    and $\omega_{2} = 3$ with initial datum in a neighbourhood of 
    the fixed point \eqref{eq:fixedck}.
}
\label{fig:anis}
\end{figure}

The above remarks imply that for $\boldsymbol{\omega}\in\omega \Z^{N}$,
the system \eqref{eq:dLanis} is an \emph{integrable, but not a MS
discretisation} of the anisotropic SW system \eqref{eq:Hlinanis}.
At the same time, the discussion made suggests that there \emph{might be}
some MS subcases of the system \eqref{eq:dLanis} for given values
of the parameters $c_{k}$. The search for this kind of MS system is
outside the possibilities given by the coalgebra approach and will be
the subject of future research. Here we limit ourselves to noting that in
\cite{GubLatDrach} the conditions for the MS of \cite{Evans2008} for
the anisotropic SW oscillator \eqref{eq:Hlinanis} were derived using a
perturbative multiple scales approach. This hints at the possibility of
a similar procedure to isolate the MS subcases using, for instance,
the multiple scales method for discrete-time systems in the spirit of
\cite{HallLustri2016}.

\subsection{Degree two invariant}

We consider the system \eqref{eq:f2} in the realisation of
$\Sl_{2}(\R)$ \eqref{eq:canrel} with canonical coordinates 
$(\vec{q}(t),\vec{p}(t) )$. The equation assumes the following form:
\begin{equation}
    \begin{aligned}
        &\phantom{+}
        q_{k}(t+h)
        \sqrt{1-\frac{b_{k}}{q_{k}^{2}(t)q_{k}^{2}(t+h)}}
        \\
        &+
        q_{k}(t-h)
        \sqrt{1-\frac{b_{k}}{q_{k}^{2}(t)q_{k}^{2}(t-h)}}
        =
        \frac{\lambda_{1} q_{k}}{\lambda_{3}-\lambda_{2}\sum_{l=1}^{N}q_{l}^{2}(t)}.
    \end{aligned}
    \label{eq:dELsl2quad}
\end{equation}
From formula \eqref{eq:pkgen} the symplectic form of the system is:
\begin{subequations}
    \begin{gather}
        q_{k}(t+h)
        \sqrt{1-\frac{b_{k}}{q_{k}^{2}(t)q_{k}^{2}(t+h)}}
        +
        p_{k}(t)
        =
        \frac{\lambda_{1} q_{k}}{\lambda_{3}-\lambda_{2}\sum_{k=1}^{N}q_{k}^{2}(t)}.
        \label{eq:eqHamqsl2quad}
        \\
        p_{k}(t+h) =
        q_{k}(t)\sqrt{1-\frac{b_{k}}{q_{k}^{2}(t)q_{k}^{2}(t+h)}}.
        \label{eq:eqHampsl2quad}
    \end{gather}
    \label{eq:eqHamsl2quad}%
\end{subequations}

\begin{remark}
    We remark that if $\lambda_{3}\neq0$ and $b_{k}=0$ for all
    $k=1,\dots,N$, this system is a generalisation of the McMillan
    map \cite{McMillan1971} introduced in \cite{McLachlan1993}.
    Indeed in such a case if $\lambda_{3}\neq 0$ considering the
    scaling $q_{k} \mapsto \sqrt{\abs{\lambda_{3}}} q_{k}$, and
    $\lambda_{1}=\lambda_{3}\sigma$ we obtain:
    \begin{equation}
        q_{k}(t+h)
        +
        q_{k}(t-h)
        =
        \frac{\sigma q_{k}}{1\mp\sum_{k=1}^{N}q_{k}^{2}(t)}.
        \label{eq:dELsl2mcm}
    \end{equation}
    We discussed the coalgebra symmetry properties in
    \cite[Subsection 6.2]{GLT_coalgebra}, where the method
    was used to prove that the system is QMS.  We finally
    recall that the system \eqref{eq:dELsl2mcm} is an example
    of $N$ d.o.f. integrable system in \emph{standard form}
    \cite{Suris1989,HietarintaJoshiNijhoff2016}, which was later extended
    in a series of papers \cite{Suris1994Garnier,Suris1994Symmetric}.
    For a general review on $N$ d.o.f. integrable systems in
    standard form we refer also to \cite[Chapter 25]{Suris2003book}.
    \label{rem:historymcm}
\end{remark}

To complete the study on this system, including continuum limits
and relation to known discrete-time integrable systems we distinguish
the cases $\lambda_{3}\neq0$ and $\lambda_{3}=0$.

\paragraph{Case $\lambda_{3}\neq 0$}

For $b_{k}\neq0$ for some $k\in\left\{ 1,2,\dots,N
\right\}$ and $\lambda_{3}\neq0$ we can scale
away this constant through the reparametrisation
$(\lambda_{1},\lambda_{2})=(\lambda_{3}\sigma_{1},\lambda_{3}\sigma_{2})$
and obtain the system:
\begin{equation}
    \begin{aligned}
        &\phantom{+}
        q_{k}(t+h)
        \sqrt{1-\frac{b_{k}}{q_{k}^{2}(t)q_{k}^{2}(t+h)}}
        \\
        &+
        q_{k}(t-h)
        \sqrt{1-\frac{b_{k}}{q_{k}^{2}(t)q_{k}^{2}(t-h)}}
        =
        \frac{\sigma_{1} q_{k}}{1-\sigma_{2}\sum_{l=1}^{N}q_{l}^{2}(t)}.
    \end{aligned}
    \label{eq:dELsl2quada3_1}
\end{equation}

\begin{remark}
    We remark that the discrete-time system \eqref{eq:dELsl2quada3_1}
    is symmetric with respect to the coordinate transformation
    \eqref{eq:dsymm}, thus allowing us to extend the domain of
    definition for negative values of the coordinates. However, from
    the right hand side of equation \eqref{eq:dELsl2quada3_1} we might
    need to restrict the domain of definition of this discrete-time
    system. Indeed, if $\sigma_{2}>0$ we need to exclude the $N$-sphere
    $\mathbb{S}^{N}(\sigma_{2}^{-1})=\left\{ \vec{q}(t) = \sigma_{2}^{-1}
    \right\}$.  So, in the end we have that the domain of definition of
    the discrete-time system \eqref{eq:dELsl2quada3_1} is:
    \begin{equation}
        \widehat{U} =
        \begin{cases}
            U \cap \mathbb{S}^{N}(\sigma_{2}^{-1}) & \text{if $\sigma_{2}>0$},
            \\
            U & \text{if $\sigma_{2}<0$}.
        \end{cases}
        \label{eq:domdQMSWoj}
    \end{equation}
    where the set $U$ is given in equation \eqref{eq:setU}.
    \label{rem:dQMSWojdef}
\end{remark}

In this case the coalgebraic invariant \eqref{eq:I2} has the following 
explicit form\footnote{We remove the common factor $\lambda_{3}$ and add 
a cosmetic $1/2$ factor to mimic the continuum case.}:
\begin{equation} 
    H_{2}\left( \lambda_{3}\neq0 \right) = \frac{1}{2}\sum_{k=1}
        \left( p_{k}^{2}  + q_{k}^{2}+ \frac{b_{k}}{q_{k}^{2}}
        -\sigma_{1} q_{k}p_{k}\right)
        -\frac{\sigma_{2}}{2} \left(\sum_{l=1}^{N} q_{l}^{2}  \right)^2,
    \label{eq:H2}
\end{equation}
From the coalgebra construction, we get the following two sets of invariants
\begin{subequations}
    \begin{align}
        \mathcal{L}_{2} &= 
        \left\{ H_{2}\left( \lambda_{3}\neq0 \right),\sC^{[2]},\dots,\sC^{[N]} \right\},
        \\
        \mathcal{R}_{2} &= 
        \left\{ H_{2}\left( \lambda_{3}\neq0 \right),\sC_{[2]},\dots,\sC_{[N]} \right\}.
    \end{align}
    \label{eq:invsetsquad}%
\end{subequations}
By induction on the d.o.f. $N$, it is easy to see that both
sets are functionally independent. As expected, this implies that the
system \eqref{eq:dELsl2quada3_1} is Liouville integrable. Similarly the set 
\begin{equation}
    \mathcal{I}_{2} = \left\{ H_{2}\left( \lambda_{3}\neq0 \right),\sC^{[2]},\dots,\sC^{[N]},
    \sC_{[2]},\dots,\sC_{[N-1]} \right\},
    \label{eq:fullinvsetsquad}
\end{equation}
is made of functionally independent functions which implies the system
\eqref{eq:dELsl2quada3_1} is QMS.

Coming to the continuum limit, we see that applying the scaling
\eqref{eq:scal} as $h\to0$ we obtain from the dEL equations
\eqref{eq:dELsl2quada3_1}:
\begin{equation}
    h^{3} \left[ \ddot{Q}_{k} + \omega^{2} Q_{k} -\frac{\beta_{k}}{Q_{k}^{3}} 
    + 2 \sigma_{2} Q_{k} \sum_{l=1}^{N} Q_{l}^{2}    \right]+ \bigO(h^{4}) = 0,
    \label{eq:WojciechowskiQ}
\end{equation}
while for the invariant \eqref{eq:H2} we have:
\begin{equation}
    H_{2} = h^{4} H_\text{W}^\text{QMS}+\bigO(h^{5}),
    \quad
    H_\text{W}^\text{QMS} = 
    \frac{1}{2}\sum_{k=1}^{N} 
    \left( P_{k}^{2}+\omega^{2}Q_{k}^{2} +\frac{\beta_{k}}{Q_{k}^{2}} \right)
    +\frac{\sigma_{2}}{2} \left( \sum_{k=1}^{N} Q_{k}^{2} \right)^{2}.
    \label{eq:HWoj}
\end{equation}
The classical continuum system defined by $H_\text{W}^\text{QMS}$ clearly
possesses coalgebra symmetry with respect to $\Sl_{2}(\R)$, and in fact it turns out to be a QMS deformation of the SW system:
\begin{equation} 
    H_\text{W}^\text{QMS} = 
    H_\text{SW} + \sigma_{2} F(\vec{q}^{2}),
    \quad
    F(\xi) = \frac{\xi^{2}}{2}.
    \label{eq:HWojVsHSW}
\end{equation}
The continuum first integrals are again given by formula \eqref{eq:CmHk}.

The system \eqref{eq:HWojVsHSW} is a particular case of the Wojciechowski system defined by the following Hamiltonian \cite{Wojciechowski1985}:
\begin{equation} 
    H_\text{W} = 
    \frac{1}{2}\sum_{k=1}^{N} 
    \left( P_{k}^{2}+\omega^{2}_{k}Q_{k}^{2} +\frac{\beta_{k}}{Q_{k}^{2}} \right)
    +\frac{\sigma_{2}}{2} \left( \sum_{k=1}^{N} Q_{k}^{2} \right)^{2}.
    \label{eq:HWojgen}
\end{equation}
In general, this system is proven to be integrable both through direct
construction of the integrals or by the existence of a Lax pair,
see \cite{Wojciechowski1985}.

This system was discretised in \cite{Suris1994inversesquare} as:
\begin{equation}
    \begin{aligned}
        &\phantom{+}
        q_{k}(t+h)
        \sqrt{1-\frac{b_{k}}{q_{k}^{2}(t)q_{k}^{2}(t+h)}}
        \\
        &+
        q_{k}(t-h)
        \sqrt{1-\frac{b_{k}}{q_{k}^{2}(t)q_{k}^{2}(t-h)}}
        =
        \frac{2c_{k} q_{k}}{1+\sum_{l=1}^{N}c_{l}q_{l}^{2}(t)}.
    \end{aligned}
    \label{eq:suriswoj}
\end{equation}
Clearly, this system is a generalisation of the system
\eqref{eq:dELsl2quad} with $\lambda_{3}\neq 0$, using the identification
$c_{k}=\sigma_{2}/2$ and making the proper scaling in the coordinates
$q_{k}$ and the parameters $b_{k}$\footnote{It is needed to go through
the original system \eqref{eq:dELsl2quad} and properly scale the
parameters.}.  So, we have that the system \eqref{eq:dELsl2quada3_1}
is a \emph{coalgebraic QMS subcase of the discrete-time Wojciechowski system
\eqref{eq:suriswoj}}.

\begin{remark}
    We remark that \emph{there might be} superintegrable
    subcases of the Wojciechowski system \eqref{eq:HWojgen} if
    $\omega_{i}/\omega_{j}\in\Q$ for some $i,j\in\{1,\dots,N\}$.
    At present, we did not prove the existence of these superintegrable
    subcases, but we limit ourselves to notice that the proofs of
    integrability of both \cite{Suris1994inversesquare,Wojciechowski1985}
    do not work in the QMS case.  So, it is reasonable to believe
    that the presented first integrals are not exhaustive of all the
    integrable cases and there might be intermediate cases between the
    Liouville integrable case ($N$ invariants) and the QMS case ($2N-2$
    invariants).  
    The problem of the
    existence of superintegrable subcases of the Wojciechowski system
    \eqref{eq:HWojgen} and its (possible) superintegrable discretisation
    will be the subject of future research.
    \label{rem:suris}
\end{remark}

\paragraph{Case $\lambda_{3} = 0$}
Let us assume now that $\lambda_{3}=0$ and that $b_{k}$ are arbitrary
(possibly also zero for all $k= 1,2,\dots,N$).  In such a case, we can
rescale $\lambda_{1}=-\lambda_{2}\sigma$ and the system
\eqref{eq:dELsl2quad} reduces to:
\begin{equation}
    \begin{aligned}
        &\phantom{+}
        q_{k}(t+h)
        \sqrt{1-\frac{b_{k}}{q_{k}^{2}(t)q_{k}^{2}(t+h)}}
        \\
        &
        +
        q_{k}(t-h)
        \sqrt{1-\frac{b_{k}}{q_{k}^{2}(t)q_{k}^{2}(t-h)}}
        =
        \sigma q_{k}\left({\sum_{l=1}^{N}q_{l}^{2}(t)}\right)^{-1}.
    \end{aligned}
    \label{eq:dELsl2quada3_0}
\end{equation}

\begin{remark}
    We remark that the discrete-time system \eqref{eq:dELsl2quada3_0}
    is symmetric with respect to the coordinate transformation
    \eqref{eq:dsymm}, thus allowing us to extend the domain of
    definition for negative values of the coordinates. In this case the
    right hand side does not give any additional restriction, so the
    system can be defined on the whole set $U$ given in equation \eqref{eq:setU}.
    \label{rem:degendeg2def}
\end{remark}

The coalgebraic invariant \eqref{eq:I2} has the following 
explicit form\footnote{We remove the common factor $\lambda_{2}$ and add 
a cosmetic $1/2$ factor.}:
\begin{equation} 
    H_{2}(\lambda_{3}=0) = \frac{\sigma}{2}\sum_{k=1}^{N}q_{k}p_{k}
    +\frac{1}{2} \left(\sum_{l=1}^{N} q_{l}^{2}  \right)^2,
    \label{eq:H2a3_0}
\end{equation}

From the coalgebra construction, we get the same two sets of
invariants as for $\lambda_{3}\neq0$ \eqref{eq:invsetsquad} case
with $H_{2}(\lambda_{3}\neq0)$ replaced by $H_{2}(\lambda_{3}=0)$.
The two sets are clearly functionally independent. In the same way
replacing $H_{2}(\lambda_{3}\neq0)$ with $H_{2}(\lambda_{3}=0)$ in
\eqref{eq:fullinvsetsquad} we obtain a set of $2N-2$ functionally
independent invariants for \eqref{eq:dELsl2quada3_0}. Summing up, we
proved that the system \eqref{eq:dELsl2quada3_0} is Liouville integrable
and moreover is QMS.

Differently from the $\lambda_{3}\neq0$ the continuum limit of the
case $\lambda_{3}=0$ is not known. Analogously to \cite[Subection
5.3]{GLT_coalgebra} it is possible to prove that no scaling of the form
\begin{equation}
    \vec{q}(t) = \vec{q}_{0}+ A h^\gamma \vec{Q}(t),
    \quad
    \sigma = \sigma(h), 
    \quad
    \vec{q}_{0}\in\Sph^{N},
    \quad \gamma\in\N,
    \label{eq:climN}
\end{equation}
where $\sigma(h)$ is an analytic function of its argument, balances the 
terms in the systems \eqref{eq:dELsl2quada3_0}.

\subsection{Degree three invariant}
We consider the system \eqref{eq:f3} in the realisation of
$\Sl_{2}(\R)$ \eqref{eq:canrel} with canonical coordinates 
$(\vec{q}(t),\vec{p}(t) )$. The equation assumes the following form:
\begin{equation}
    \begin{aligned}
        &\phantom{+}
        q_{k}(t+h)
        \sqrt{1-\frac{b_{k}}{q_{k}^{2}(t)q_{k}^{2}(t+h)}}
        \mp q_{k}(t)
        \\
        &+
        q_{k}(t-h)
        \sqrt{1-\frac{b_{k}}{q_{k}^{2}(t)q_{k}^{2}(t-h)}}
        =
        \tau q_{k}\left(\sum_{l=1}^{N}q_{l}^{2}(t)\right)^{-1}.
    \end{aligned}
    \label{eq:dELsl2cub}
\end{equation}
From formula \eqref{eq:pkgen} the symplectic form of the system is:
\begin{subequations}
    \begin{gather}
        q_{k}(t+h)
        \sqrt{1-\frac{b_{k}}{q_{k}^{2}(t)q_{k}^{2}(t+h)}}
        \mp q_{k}(t)
        +
        p_{k}(t)
        =
        \tau q_{k}\left(\sum_{l=1}^{N}q_{l}^{2}(t)\right)^{-1},
        \label{eq:eqHamqsl2cub}
        \\
        p_{k}(t+h) =
        q_{k}(t)\sqrt{1-\frac{b_{k}}{q_{k}^{2}(t)q_{k}^{2}(t+h)}}.
        \label{eq:eqHampsl2cub}
    \end{gather}
    \label{eq:eqHamsl2cub}%
\end{subequations}

\begin{remark}
    We remark that if $b_{k}=0$ for all $k=1,\dots,N$, the system
    \eqref{eq:dELsl2cub} with the plus sign
    \begin{equation}
        q_{k}(t+h)
        +q_{k}(t)
        +
        q_{k}(t-h)
        = \tau q_{k}\left(\sum_{k=1}^{N}q_{k}^{2}(t)\right)^{-1}.
        \label{eq:dELsl2dPI}
    \end{equation}
    is a particular case of an autonomous version of the discrete-time Painlev\`e I equation \cite{Grammaticosetal1991,HietarintaJoshiNijhoff2016} we introduced in \cite[Section 5]{GLT_coalgebra} with the parameter $\vbeta = \vec{0}$.
    The coalgebra symmetry properties of the system \eqref{eq:dELsl2dPI}
    were discussed in \cite{GLT_coalgebra}, where we showed that the system is QMS. We also recall that the continuum limit of the system
    \eqref{eq:dELsl2dPI} is unknown.
    \label{rem:historyNdPI}
\end{remark}

Considering $b_{k}\neq0$ for some $k\in\left\{ 1,\dots,N \right\}$
the coalgebraic invariant \eqref{eq:I3} has the following 
explicit form\footnote{We added a cosmetic $1/2$ factor.}:
\begin{equation} 
    H_{3} = 
    \frac{1}{2}
    \sum_{k=1}^{N}\left[ \left(p_{k} \mp q_{k}\right)^{2}+ \frac{b_{k}}{q_{k}^{2}}\right]
    \left[\left(\tau- \sum_{k=1}^{N}q_{k}p_{k} \right)^{2} +\sC^{[N]}\right].
    \label{eq:H3}
\end{equation}
From the coalgebra construction, we get the following two sets of invariants
\begin{equation}
    \mathcal{L}_{3} = 
    \left\{ H_{3},\sC^{[2]},\dots,\sC^{[N]} \right\},
    \quad
    \mathcal{R}_{3} = 
    \left\{ H_{3},\sC_{[2]},\dots,\sC_{[N]} \right\}.
    \label{eq:invsetscub}
\end{equation}
By induction on the d.o.f. $N$ it is easy to see that both
sets are functionally independent. As expected, this implies that the
system \eqref{eq:dELsl2cub} is Liouville integrable. Similarly the set 
\begin{equation}
    \mathcal{I}_{3} = \left\{ H_{3},\sC^{[2]},\dots,\sC^{[N]},
    \sC_{[2]},\dots,\sC_{[N-1]} \right\},
    \label{eq:fullinvsetscub}
\end{equation}
is made of functionally independent functions which implies the system
\eqref{eq:dELsl2cub} is QMS.

Based on the above discussion we have that the system \eqref{eq:dELsl2cub}
for $b_{k}\neq0$ for some $k\in\left\{ 1,\dots,N \right\}$ is \emph{a novel
QMS system, generalising the already known system \eqref{eq:dELsl2dPI}}.

We finally note that at present no continuous analogue of the system
\eqref{eq:dELsl2cub} is known. Indeed, analogously to the system
\eqref{eq:dELsl2dPI} mentioned in Remark \ref{rem:historyNdPI} and to
\cite[Subection 5.3]{GLT_coalgebra}, it is possible to prove that no
scaling of the form \eqref{eq:climN} with $\tau(h)$ analytic functions
of its argument balances the terms in the systems \eqref{eq:dELsl2cub}.

\section{Conclusions}
\label{sec:concl}

In this paper, we gave the necessary and sufficient conditions for a discrete-time system
in quasi-standard form \eqref{eq:dELgen} to admit coalgebra symmetry
with respect to the generic symplectic realisation of the Lie--Pois\-son
algebra $\Sl_{2}(\R)$ \eqref{eq:sl2poisson} in $N$ degrees of freedom.
In particular, from Theorem \ref{thm:rad} we see that the most general
form of this system is the one considered in \cite{Suris1994inversesquare}.
In this sense, we see that the coalgebra symmetry approach naturally selects
the functional form for the functions $\ell_{k}=\ell_{k}(\xi)$, $k=1,\dots,N$.

As already discussed in \cite{GLT_coalgebra} the systems in
quasi-stand\-ard form \eqref{eq:dELgen} admitting coalgebra symmetry
with respect to the generic symplectic realisation Lie--Pois\-son algebra
$\Sl_{2}(\R)$ \eqref{eq:sl2poisson} in $N$ degrees of freedom naturally
possess $2N-3$ functionally independent invariants.  However, in general
these systems are not Liouville integrable. This is because we are missing
an exact discrete-time equivalent of the Hamiltonian. So, to characterise
the Liouville integrable cases we searched for polynomial invariants
\emph{in the variables of the algebra $\Sl_{2}(\R)$} of increasing degree
for the associated dynamical systems \eqref{eq:sl2evol}. It turns out
that such invariants do exist for degrees 1, 2, and 3, while for degrees
4, and 5 no new systems arise. This led us to conjecture that the only
non-trivial Liouville integrable systems with a polynomial invariant are
those admitting an invariant of degrees 1, 2, and 3. The proof of this
conjecture is particularly challenging because, in general, it does not
seem possible to resum the expression \eqref{eq:Idtrasl} to:
\begin{equation}
    I_{d}(t+h) = \sum_{l=0}^{d} F_{l}(J_{-}) J_{+}^{l}J_{3}^{d-l},
    \label{eq:Idwished}
\end{equation}
which will make the claim easier to prove.  Unfortunately,
since in general not even the degree of the function $f=f(\xi)$
is known it is not possible to apply the ``destructive'' test
of algebraic entropy \cite{BellonViallet1999, GubbiottiASIDE16,
GrammaticosHalburdRamaniViallet2009}\footnote{Quoting from
\cite{Viallet2021_Levi70} ``\emph{Algebraic entropy.} pro: it is canonical
(invariant by birational changes of coordinates), and the vanishing of
the entropy may serve as a characterisation of integrability, as the
sign of catastrophic drop of the complexity, con: destructive rather
than constructive, since it gives a yes/no answer to the question
“is this model integrable?”''. Italics from the original text.}.
For the very same reason it is not possible to apply the ``constructive''
singularity confinement test \cite{GrammaticosHalburdRamaniViallet2009,
Grammaticosetal1991}. Under the (reasonable) assumption that $f(\xi)$
is a rational function, one could try to apply a Nevanlinna theory for
maps \cite{Ablowitz_et_al2000}. However, at present, such a theory is
not enough developed to treat a system of the form \eqref{eq:sl2evol}.
So, this topic will be the subject of further research.

In the sequent section we carefully analysed the obtained systems in
the realisation \eqref{eq:sl2poisson} and proved explicitly their
Liouville integrability and quasi-maximal superintegrability. In
addition, we found a maximally superintegrable system given by the
discrete Lagrangian \eqref{eq:dELsl2lin} which is a discretisation of
the well-known SW system \cite{Fris1965,Evans1990}. We also discussed
the possible generalisation \eqref{eq:dLanis} discretising the so-called
caged anisotropic oscillator \cite{Evans2008}, which we showed to be
Liouville integrable. Unfortunately, differently from its continuous
counterpart, the system \eqref{eq:dLanis}does not appear to be maximally
superintegrable from a numerical study of its orbits. Since this kind
of analysis goes outside the applicability of the coalgebra method,
this topic will be the subject of future research. Besides these cases,
we found the coalgebraic subcase of the discrete-time Wojciechowski
system \eqref{eq:dELsl2quad}, whose general case was constructed in
\cite{Suris1994inversesquare}. Finally, we found a non-birational
generalisation of a system we proposed in \cite{GLT_coalgebra}, which
we deem to be new. In \Cref{tab:intcases} we give a compendium of the
known Liouville integrable cases of equation \eqref{eq:dELgen}.

Summing up, in this paper we proved that the coalgebra symmetry method
can be applied to systematically produce $N$ d.o.f. discrete-time
superintegrable systems.

An interesting open problem is the existence of a coalgebraic 
discretisation of the Kepler--Coulomb system:
\begin{equation}
    H_\text{KC} = 
    \frac{1}{2} \sum_{i=1}^{N} p_{i}^{2} 
    -\alpha \left( \sum_{i=1}^{N} q_{i}^{2}  \right)^{-1/2}.
    \label{eq:HKC}
\end{equation}
This model is MS, through the existence of an additional
integral of motion called the Laplace--Runge--Lenz vector
\cite{LeachFlessas2003,MillerPostWinternitz2013R}. It would be
interesting to show if it is possible to construct such an invariant in
the discrete-time setting.

Another open problem is the generalisation to non-Euclidean
manifolds.  For the coalgebra approach to these systems,
see for example a series of papers by Italian--Spanish school
\cite{Ballesteros_et_al2011AnnPhys,Ballesteros_et_al2009AnnPhys,
Ballesteros_et_al2008PhysAtomNuclei,BallesterosHerranz2007,
Ballesteros_et_al2008PhysD} culminating in the proof of a Bertrand-like
theorem on curved space \cite{Ballesteros:2009xqz}, linked to the
so-called Perlick classification \cite{Perlick1992}.  For a different,
more geometric perspective on the subject, see the recent classification
in \cite{Kress_etal2019, Kress_etal2020}.  Finding a connection between
these two approaches and building the discrete-time analogue will be
subject of future research.

Finally, we note that the present construction can be applied to other
Lie--Poisson algebras, especially the ones related to classified Lie
algebras for which the Casimir functions are known, see \cite{sno14,
Montreal1976,Montreal1977,Mubarakzyanov1963a,Mubarakzyanov1963b,
Mubarakzyanov1963c,Mubarakzyanov1966}. In the continuous setting this
was done by Ballesteros and Blasco \cite{Ballesteros_2008, BlascoPhD}. We
note that the most natural extension would be the two-photon or $h_{6}$
Lie--Poisson algebra \cite{Zhang_et_al1990}, a Lie--Poisson algebra
containing many other interesting Lie--Poisson algebras as subalgebras,
including $\Sl_{2}(\R)$, whose associated Hamiltonian integrable systems
have been discussed in \cite{BallesterosHerranz2001,BlascoPhD}.

\begin{table}
    \centering
    \begin{tabular}{llp{3.7cm}p{3.7cm}}
        \toprule
        $\deg I$ & $V(\xi)$ & Standard form & Quasi-standard form
        \\
        \midrule
        1 & $\displaystyle  \frac{\alpha_{1}\xi}{2}$ & 
        Discrete-time isotropic harmonic oscillator \cite{GLT_coalgebra} & 
        \emph{Discrete-time SW system \eqref{eq:dELsl2lin}}
        \\
        2 & $\displaystyle-\frac{\alpha_{1}}{2\alpha_{2}}\log\left( 1- \alpha_{2}\xi \right)$ & 
        $N$ d.o.f. McMillan system \cite{McLachlan1993,Suris1994Garnier} & 
        QMS Wojciechowski system \eqref{eq:dELsl2quada3_1} \cite{Suris1994inversesquare}
        \\
        2 & $\displaystyle\frac{\alpha_{1}}{2}\log \xi $ & 
        \emph{New QMS system \eqref{eq:dELsl2quada3_0}} & 
        \emph{New QMS system \eqref{eq:dELsl2quada3_0}} 
        \\
        3 & $\displaystyle\pm\frac{\xi}{2}+\frac{\alpha_{1}}{2}\log\xi$ & 
        $N$ d.o.f. aut-dPI system \eqref{eq:dELsl2dPI} \cite{GLT_coalgebra} & 
        \emph{New QMS system \eqref{eq:dELsl2quada3_0}} 
        \\
        \bottomrule
    \end{tabular}
    \caption{Known integrable cases of equation \eqref{eq:dELsl2}.
    The novel systems are highlighted in italic.}
    \label{tab:intcases}
\end{table}

\section*{Acknowledgements}

This work was made in the framework of the Project ``Meccanica dei
Sistemi discreti'' of the GNFM unit of INDAM. In particular, GG
acknowledge support of the GNFM through Progetto Giovani GNFM 2023:
``Strutture variazionali e applicazioni delle equazioni alle differenze
ordinarie'', CUP\_E53C22001930001.

DL was supported by the Australian Research Council Discovery Project
DP190101529 (A/Prof. Y.-Z. Zhang).

The figures in this paper are \texttt{pdf} and were produced 
in \texttt{python} using the libraries \texttt{numpy} \cite{Harris2020NumPy} 
and \texttt{mathplotlib} \cite{Hunter2007}.

\appendix

\section{Explicit form of the coefficients in the case $d=3$}
\label{app:expl}

We list the explicit formulas for the coefficients obtained in the case
$d=3$ during the proof of Theorem \ref{thm:invariant}.
We divide the polynomial $I_{3}$ into its homogeneous components:
\begin{equation}
    I_{3}(t) = I_{3}^{(1)}(t) + I_{3}^{(2)}(t) + I_{3}^{(3)}(t).
    \label{eq:I3hom}
\end{equation}
and list the coefficients by the degree of the component.

Degree three:
\begin{equation}
    \begin{gathered}
        a_{{0,0,3}}=a_{{0,1,2}}=a_{{0,2,1}}=a_{{0,3,0}}=a_{{1,0,2}}=
        a_{{2,0,1}}=0,a_{{3,0,0}}=0,
        \\
        a_{{2,1,0}}=a_{{1,2,0}},
        \quad
        a_{{1,1,1}}=-2 a_{{1,2,0}},
        \quad
        a_{{1,2,0}}=a_{{1,2,0}}    
    \end{gathered}
    \label{eqI3plus3}
\end{equation}

Degree two:
\begin{equation}
    a_{{2,0,0}}=a_{{0,2,0}}=0,
    \quad
    a_{{1,0,1}}=a_{{0,1,1}},
    \quad
    a_{{0,0,2}}=-2 a_{{0,1,1}}-a_{{1,1,0}}.
    \label{eqI3plus2}
\end{equation}

Degree one:
\begin{equation}
    a_{{0,0,1}}=\mp\frac{1}{2} {\frac {{a^{2}_{{0,1,1}}}}{a_{{1,2,0}}}},
    \quad
    a_{{0,1,0}}= \frac{1}{4} {\frac {{a^{2}_{{0,1,1}}}}{a_{{1,2,0}}}},
    \quad
    a_{{1,0,0}}= \frac{1}{4} {\frac {{a^{2}_{{0,1,1}}}}{a_{{1,2,0}}}}.
    \label{eqI3plus1}
\end{equation}

\section{Periodic cases of the system \eqref{eq:sl2evol} for $V=\kappa\xi/2$}
\label{app:per}

If $V=\kappa\xi/2$ then the system \eqref{eq:sl2evol} becomes:
\begin{subequations}
    \begin{align}
        J_{+}(t+h) &= J_{-}(t),
        \label{eq:Jpevolp}
        \\
        J_{-}(t+h) &= J_{+}(t) - 2\kappa J_{3}(t)
        +\kappa^{2}J_{-}(t),
        \label{Jmevolp}
        \\
        J_{3}(t+h) &= -J_{3}(t) +\kappa J_{-}(t).
        \label{J3evolp}
    \end{align}
    \label{eq:sl2evolp}
\end{subequations}
This system is linear and it can be written in matrix form as follows:
\begin{equation}
\vec{J}(t+h) = M \vec{J}(t),
\label{eq:sl2evolpM}
\end{equation}
where $\vec{J}(t) = (J_{+}(t), J_{-}(t), J_{3}(t))$ and
\begin{equation}
    M=
    \begin{pmatrix}
        0&1&0
        \\ 
        1&\kappa^{2}&-2\kappa
        \\
        0&\kappa&-1
    \end{pmatrix}.
    \label{eq:matrixM}
\end{equation}
From standard techniques in systems of linear difference equations
\cite[Chap. 3]{Elaydi2005} we have that the solution is given by:
\begin{equation}
    \vec{J}(t_{0} + lh) = M^{l} \vec{J}(t_{0}).
    \label{eq:sl2evolpsol}
\end{equation}
Still following \cite[Chap. 3]{Elaydi2005} we can compute the $l$th
power of the matrix $M$ \eqref{eq:matrixM} through diagonalisation or Jordan block-form reduction. The characteristic polynomial of
\eqref{eq:matrixM} is:
\begin{equation}
    p_{M} (\mu)
    =
    (\mu-1)\left[ \mu^2+(2-\kappa^2) \mu+1\right].
    \label{eq:pM}
\end{equation}
Hence, the eigenvalues are:
\begin{equation}
    \mu_{1} = 1,
    \quad
    \mu_{2} = 
    \frac{\kappa^2}{2}-1 + \frac{\kappa}{2} \sqrt{\kappa^2-4}, 
    \quad
    \mu_{3} = 
    \frac{\kappa^2}{2}-1 - \frac{\kappa}{2} \sqrt{\kappa^2-4}.
    \label{eq:Meigen}
\end{equation}
The eigenvalues are different for all $\kappa\neq \pm2$.
This readily implies that the matrix is diagonalisable for every $\kappa\neq \pm2$.
When $\kappa=\pm2$ the matrix is not diagonalisable because the only
eigenvalue is $\mu_{1}=1$, which has geometric multiplicity one.

Then for $\kappa\neq\pm2$ we have:
\begin{equation}
    M^{l} = G \diag\left(1,
        \left(\frac{\kappa^2}{2}-1 + \frac{\kappa}{2} \sqrt{\kappa^2-4}\right)^{l},
        \left(\frac{\kappa^2}{2}-1 - \frac{\kappa}{2} \sqrt{\kappa^2-4}\right)^{l}
    \right)G^{-1}.
    \label{eq:Mlth}
\end{equation}
The periodicity condition is that there exists a $L\in\N$ such that
$M^{L}=\mathbb{I}_{3}$, where $\mathbb{I}_{3}$ is the identity matrix.
From equation \eqref{eq:Mlth} this is true if and only if:
\begin{equation}
    \left(\frac{\kappa^2}{2}-1 + \frac{\kappa}{2} \sqrt{\kappa^2-4}\right)^{L}=1,
    \quad
    \left(\frac{\kappa^2}{2}-1 - \frac{\kappa}{2} \sqrt{\kappa^2-4}\right)^{L}=1.
    \label{eq:condper}
\end{equation}
Since the second equation in \eqref{eq:condper} is the same as the first
up to the discrete symmetry $\kappa\mapsto -\kappa$, we can just 
consider the solution of the first equation in \eqref{eq:condper}.
That is, the solutions of \eqref{eq:sl2evolpM} are periodic of period $L\in\N$
if and only if $\kappa$ satisfies the following algebraic equation:
\begin{equation}
    \left(\frac{\kappa^2}{2}-1 + \frac{\kappa}{2} \sqrt{\kappa^2-4}\right)^{L}=1.
    \label{eq:condper2}
\end{equation}
That is, the left hand side of equation \eqref{eq:condper2} must be
a $L$th root of unity. Recalling that $L$th roots of 
unity can be written in complex exponential form as 
$z_{k}= \exp(2 \imath k\pi/L)$ with $k=0,\dots,L-1$ we have:
\begin{equation}
    \frac{\kappa^2}{2}-1 + \frac{\kappa}{2} \sqrt{\kappa^2-4}=\exp\left( \frac{2 \imath k\pi }{L} \right),
    \quad
    k=0,\dots,L-1.
    \label{eq:condper2L}
\end{equation}
That is, the solution can be written as:
\begin{equation}
    \kappa = \pm2 \cos\left( \frac{k\pi}{L} \right),
    \quad
    k=0,\dots,L-1.
    \label{eq:condper2Lsol}
\end{equation}
However, this is not definitive: we need to discard $k=0$ because
it yields $\kappa=\pm2$ which is not acceptable, and since the cosine
function is antiperiodic of antiperiod $\pi$ we can choose the sign plus
in \eqref{eq:condper2Lsol}. So we denote the final expression of the 
solutions as  $\kappa_{k,L}$ and its expression is:
\begin{equation}
    \kappa_{k,L} = 2 \cos\left( \frac{k\pi}{L} \right),
    \quad
    k=1,\dots,L-1.
    \label{eq:condper2Lsolfin}
\end{equation}

We underline that from formula \eqref{eq:condper2Lsolfin} we have
$\abs{\kappa_{k,L}}<2$, and for all values of $k$ and $L$ except for
\begin{equation}
    \frac{k}{L} =  \frac{1}{3}, \frac{1}{2}, \frac{2}{3},
    \label{eq:kLrat}
\end{equation}
the numbers $\kappa_{k,L}$ are irrational numbers. In particular, this
implies that, for every $L$ prime we will obtain new solutions. This
implies that the set of values of $\kappa_{k,L}$ such that the dynamical
system is periodic is a countably infinite set.

\begin{remark}
    We remark that for $L=3$ we have:
    \begin{equation}
        \kappa_{0,3} = 1,
        \quad
        \kappa_{1,3} = -1.
        \label{eq:kappak3}
    \end{equation}
    That is, these cases correspond to the degenerate case of the 
    function $f_{3}$ as $\tau\to0$, see \Cref{fig:scheme}.
    \label{rem:L3}
\end{remark}

\printbibliography

\end{document}